\documentclass[aps,pra,onecolumn,superscriptaddress,tightenlines,floatfix,12pt]{revtex4}

\renewcommand\thesection{\arabic{section}}
\renewcommand\thesubsection{\arabic{section}.\arabic{subsection}}
\makeatletter
\renewcommand\p@subsection{}
\makeatother

\makeatletter
\renewcommand\p@subsubsection{}
\makeatother

\usepackage[a4paper,left=2cm,right=2cm,top=3cm,bottom=3cm]{geometry}

\usepackage{graphicx}
\usepackage{amsmath}
\usepackage{amssymb}
\usepackage{amsthm}
\usepackage{amsbsy}
\usepackage{mathrsfs}
\usepackage{varioref}
\usepackage{dsfont}
\usepackage{bm}
\usepackage{color}
\usepackage[colorlinks = false, citecolor=green, linkcolor=red, breaklinks=true, linktocpage=true]{hyperref}
\usepackage[footnotesize,center,bf]{caption}
\usepackage{subfigure}

\usepackage{booktabs} % for much better looking tables
\usepackage{array} % for better arrays (eg matrices) in maths
\usepackage{paralist} % very flexible & customisable lists (eg. enumerate/itemize, etc.)
\usepackage{verbatim} % adds environment for commenting out blocks of text & for better verbatim
\usepackage[parfill]{parskip} % Activate to begin paragraphs with an empty line rather than an indent
\newcommand{\eq}[1]{Eq.~\eqref{#1}}
\newcommand{\fig}[1]{Fig.~(\ref{#1})}
\newcommand{\sect}[1]{Sec.~\ref{#1}}
\newcommand{\app}[1]{Appendix~(\ref{#1})}
\newcommand{\lemref}[1]{Lemma~\ref{#1}}

\newcommand{\thmref}[1]{Theorem~\ref{#1}}
\newcommand{\corref}[1]{Corollary~\ref{#1}}

\newcommand{\ts}{\textsuperscript}
\newcommand{\<}{\langle} 
\renewcommand{\>}{\rangle}
\newcommand{\ket}[1]{\left|{#1}\right\rangle}

\newcommand{\pr}[1]{|{#1}\rangle \langle {#1} |}

\newcommand{\thickbar}[1]{\mathbf{\bar{\text{$#1$}}}}

\newcommand{\h}{{\mathcal{H}}}
\newcommand{\kk}{{\mathcal{K}}}
\newcommand{\rr}{{\mathcal{R}}}

\newcommand{\oo}{{\mathcal{O}}}

\renewcommand{\aa}{{\mathcal{A}}}

\newcommand{\co}{\mathds{C}}
\newcommand{\re}{\mathds{R}}

\newcommand{\one}{\mathds{1}}

\newcommand{\imag}{\mathfrak{i}}
\newcommand{\louv}{\mathscr{L}}
\newcommand{\tr}{\mathrm{tr}}

\theoremstyle{plain}

\numberwithin{defn}{section}
\newtheorem{prop}{Proposition}
\numberwithin{prop}{section}
\newtheorem{con}{Conjecture}
\numberwithin{con}{section}
\newtheorem{lem}{Lemma}
\numberwithin{lem}{section}
\newtheorem{thm}{Theorem}
\numberwithin{thm}{section}
\newtheorem{cor}{Corollary}
\numberwithin{cor}{section}

\numberwithin{equation}{section}

\begin{document}

%+Title
\title{Minimising the heat dissipation of   quantum information erasure}

 \author{M. Hamed Mohammady}
 \email[Correspondence email: ]{hamed.mohammady@lx.it.pt} 
 \affiliation{Physics of Information Group, Instituto de Telecomunica\c{c}\~oes, P-1049-001 Lisbon, Portugal} 
\author{Masoud Mohseni}
\affiliation{Research Laboratory of Electronics, Massachusetts Institute of Technology, Cambridge, MA 02139}
\affiliation{Google Inc., Venice, CA 90291} 
\author{Yasser Omar}
\affiliation{Physics of Information Group, Instituto de Telecomunica\c{c}\~oes, P-1049-001 Lisbon, Portugal}
\affiliation{CEMAPRE, ISEG, Universidade de Lisboa, P-1200-781 Lisbon, Portugal}
\affiliation{IST, Universidade de Lisboa, P-1049-001 Lisbon, Portugal}

\date{ October 7, 2015}

\begin{abstract}
Quantum state engineering and quantum computation rely on information erasure procedures that, up to some fidelity, prepare a quantum object in a pure state.  Such processes occur within Landauer'�s framework if they rely on an interaction between the object and a thermal reservoir. Landauer'�s principle dictates that this must dissipate a minimum quantity of heat, proportional to the entropy reduction that is incurred by the object, to the thermal reservoir. However, this lower bound is only reachable for some specific physical situations, and it is not necessarily achievable for any given reservoir. The main task of our work can be stated as the minimisation of heat dissipation given probabilistic information erasure, i.e., minimising the amount of energy transferred to the thermal reservoir as heat if we require that the probability of preparing the object in a specific pure state $\ket{\varphi_1}$ be no smaller than $p_{\varphi_1}^{\max}-\delta$. Here $p_{\varphi_1}^{\max}$ is the maximum probability of information erasure that is permissible by the physical context, and $\delta\geqslant 0$ the error. To determine the achievable minimal heat dissipation of quantum information erasure within a given physical context, we explicitly optimise over all possible unitary operators that act on the composite system of object and reservoir. Specifically, we characterise the equivalence class of such optimal unitary operators, using tools from majorisation theory, when we are restricted to finite-dimensional Hilbert spaces. Furthermore, we discuss how pure state preparation processes could be achieved with a smaller heat  cost than Landauer'�s limit, by operating outside of Landauer'�s framework.

\end{abstract}

\maketitle

\section{Introduction}
  
\subsection{Information erasure and thermodynamics}
In his attempt to exorcise  Maxwell's demon \citep{Maxwellian-demon, Maxwells-demon-colloquium},   Leo Szilard    conceived of an engine \citep{Szilard} composed of a box that is in thermal contact with a  reservoir at temperature $T$, and contains a single gas particle. By placing a partition in the middle of the box and determining on which side of this the particle is located, the Maxwellian demon can attach to said partition a weight-and-pulley system  so that, as the gas expands, the weight is elevated. By ensuring that the partition moves without friction, and continuously adjusting the weight to make the process quasi-static, one may fully convert $k_B T \log(2)$ units of heat energy from the gas into  work.  Here, $k_B$ is Boltzmann's constant and $\log(\cdot)$ is the natural logarithm. In order to save the second law of thermodynamics the engine must dissipate at least $k_B T \log(2)$ units of energy to the thermal reservoir as heat. While it was initially believed that this heat dissipation is due to the measurement act by the Maxwellian demon,  following the work of   Landauer, Penrose, and  Bennet \citep{Landauer, Penrose-stat-mech, Bennett-Landauer-review, Bennett-Landauer-Notes} the responsible process was identified as the  erasure of information in the demon's memory -- the logically irreversible process of assigning a prescribed value to the memory, irrespective of its prior state. That the minimum heat dissipation required to erase one bit of information cannot be any smaller than $k_B T \log(2)$ is commonly known as Landauer's principle, and said minimum quantity as Landauer's limit. In general, Landauer's principle may be encapsulated by the Clausius inequality
\begin{equation}
 \Delta Q \geqslant k_B T\Delta S, \label{Landauer-traditional}
\end{equation}    
where $\Delta Q$ is the heat dissipation to the thermal reservoir and  $\Delta S$ is the entropy reduction in the object of information erasure. 

\subsection{Thermodynamics in the quantum regime}
Recent years have been witness to  a growing interest in   thermodynamics and statistical mechanics in the quantum regime (See \citep{Goold-thermo-review, Anders-thermo-review} for a review). This has lead to  a lively debate regarding the  definition  of two central concepts in thermodynamics -- work and heat -- within the framework of quantum theory.
In classical physics, the work done during a process  is defined as the increase in useful, ordered energy. Conversely, the heat dissipated during a process is the increase in unusable, disordered energy. In   Szilard's engine,  for example,   work is   characterised as the (deterministic) elevation of a weight, and hence the increase of its gravitational potential energy. The heat dissipated, on the other hand, would be stored as  kinetic energy  in the random motion of the atoms that constitute Szilard's engine, as well as the environment. This clear distinction fails in  quantum mechanics, which is an inherently probabilistic theory. 

Broadly speaking, work may be characterised in two different ways: (i)  $\epsilon$-deterministic work \citep{Dahlsten-Inadequacy-von-Neumann,Anders-Single-shot}; and (ii)  average work \citep{thermo-open-quantum,Eisert-Operational-Work}. In either case, one may  include the work storage device -- a quantum analogue of the elevated weight in Szilard's engine -- explicitly in the formalism, such as \citep{Horodecki-limitations} and \citep{thermo-individual-quantum}. This is not always done, and one may directly examine the energy change in the system under consideration.     In the $\epsilon$-deterministic framework, the work of a process is defined as the difference in energy measurement outcomes on the system (or work storage device), observed prior and posterior to the process. The $\epsilon$-deterministic work is then the maximum value of work, thus defined, which occurs with a probability of at least $1-\epsilon$.   Meanwhile, average work is given as either the difference in expectation values of energy, or the difference in the free energies, of the system (or work storage device) observed prior and posterior to the process. The difference in average energy can be converted to the difference in free energy by subtracting the von Neumann entropy of the system, multiplied by the temperature, from its average energy. 

Definitions of heat can similarly be broadly classified into two categories: (i) where the thermal reservoir is treated \emph{extrinsically} \citep{thermo-open-quantum,Anders-Measurement-Thermodynamics}; and (ii) where the thermal reservoir is treated \emph{intrinsically} \citep{Paternostro-Landauer-collision, Modi-nonequilibrium-Landauer} . If the thermal reservoir is treated extrinsically, whereby it does not explicitly appear in the framework as a quantum system susceptible to change and examination, heat is a property of the system of interest. One may therefore define heat after having determined work -- that is to say, given the change in total energy of the system, $\Delta E$, and the work, $\Delta W$, the heat $\Delta Q$ is given by the first law of thermodynamics as $\Delta Q = \Delta E - \Delta W$. Alternatively,  Landauer's principle may be invoked to get a lower bound of heat dissipation, given that the system has undergone an entropy change of $\Delta S$.    If the thermal reservoir is treated intrinsically, on the other hand, heat can be defined as the average energy change of the reservoir itself. In other words, heat is average work pertaining to the thermal reservoir.
A thermal reservoir, considered intrinsically, is   a system that is initially uncorrelated from every other system considered, and is prepared in a Gibbs state.  We note that, from this perspective, treating the thermal reservoir with the Born Markov approximation would render it  extrinsic; this is because the state of the reservoir, in the coarse-grained picture, is assumed to never change. As such,   defining heat dissipation during a process as the average energy increase of the reservoir would lead one to conclude that no heat is dissipated at all. Indeed, the physical justification for the Born Markov approximation is that, at time-scales much shorter than that at which the system changes, the reservoir  relaxes to its equilibrium state by interacting with an unseen and, hence extrinsic, environment. If this  environment is explicitly accounted for quantum mechanically, then the total system will again evolve unitarily, and the energy increase of this environment has to also be accounted for.   

In this article, we shall adopt the view that work is the change in average energy of the system. Moreover, whenever a thermal reservoir is mentioned, we will consider it intrinsically and include it as part of the system under investigation.    The work storage device, however, is considered extrinsically: by the first law of thermodynamics we take as a priori the notion that the change in average energy of the  system -- including the reservoir if it is present -- must come from an external energy source. This total change in average energy is defined as the \emph{work done by the extrinsic work storage device}. If the total system is composed of an object and thermal reservoir, each with a well-defined Hamiltonian, then the portion of this work that is taken up by the object is called the \emph{work done on the object}, and the portion taken by the reservoir is called the \emph{heat dissipated to the reservoir}. If the total system is thermal, then the entirety of the work done by the extrinsic work storage device is defined as heat.
  
\subsection{A quantum mechanical Landauer's principle}

The surge of interest in quantum thermodynamics has included  attempts to consider Landauer's principle  quantum mechanically \citep{Piechocinska-erasure, Quantum-Landauer-Anders, Renner-quantitative-landauer, negative-entropy, Sagawa-thermodynamic-measurement, Modi-nonequilibrium-Landauer, Landauer-Quantum-Statistical}. Most notable among such efforts is that 
of  Reeb and Wolf \citep{Reeb-Wolf-Landauer}, who provide a fully quantum statistical mechanical derivation of Landauer's principle by considering the process of reducing the entropy of a quantum  object by its joint unitary evolution with a thermal reservoir. Here, they consider heat dissipation as the average energy increase of the reservoir, which is initially  in a Gibbs state and is  not correlated with the object.   For a reservoir with a Hilbert space of finite dimension $d_\rr$, they arrive at an equality form of Landauer's principle,
\begin{equation}
 \Delta Q= k_B T \left( \Delta S + I(\oo: \rr)_{\rho'}+ S(\rho_\rr'\| \rho_\rr(\beta))\right), \label{Landauer-equality}
\end{equation}
where $I(\oo: \rr)_{\rho'}$ is the mutual information between object and reservoir after the joint evolution, and $S(\rho_\rr'\| \rho_\rr(\beta))$ is the relative entropy between the post-evolution state of the reservoir and its initial state at thermal equilibrium. As the mutual information and relative entropy terms are non-negative, this implies Landauer's principle. While \eq{Landauer-equality}  always yields the exact heat dissipation, it involves terms that are cumbersome to calculate and, perhaps more importantly, it is not a function of $\Delta S$ alone. As such, Reeb and Wolf  provide an inequality form of Landauer's principle,
\begin{equation}
\Delta Q \geqslant k_B T (\Delta S + M(\Delta S, d_\rr)),\label{Landauer-corrected}
\end{equation}
where $M(\Delta S, d_\rr)$ is a non-negative correction term that vanishes in the limit as $d_\rr$ tends to infinity.  
\subsection{The need for a context-dependent  Landauer's principle}
       The study in \citep{Reeb-Wolf-Landauer}  provides a lower bound of energy transferred to the thermal reservoir as heat dissipation, given that the object's entropy decreases by $\Delta S$ and that the reservoir's Hilbert space dimension is $d_\rr$. The crucial point however is that this lower bound can be  obtained for \emph{some} physical context,  but not all of them. By physical context, we mean the tuple $(\h_\oo, \rho_\oo, \h_\rr,  H_\rr,T)$. Here $\h_\oo$ and $\rho_\oo$ are respectively the Hilbert space and state of the object, while $\h_\rr$, $H_\rr$, and $T$ are respectively the Hilbert space, Hamiltonian, and  temperature of the reservoir.   For example, one way to achieve the  lower bound of  \eq{Landauer-corrected} is for the object and reservoir to have the same Hilbert space dimension, allowing us to perform a swap map between them; this will take the mutual information term in \eq{Landauer-equality} to zero. The next step of the optimisation would be to pick a  specific  $\rho_\oo$, $H_\rr$ and $T$ so as to minimise the relative entropy term.    Conversely, for a given physical context  such inequalities may prove less instructive. Indeed,  if it is impossible to achieve  the lower bound of \eq{Landauer-corrected} in a given experimental setup,  in what sense can we consider this as the lowest possible heat dissipation due to information erasure?  In this study, therefore, we aim to approach the problem of information erasure from the dual perspective: given a physical context, what is the minimum heat that must be dissipated in order to achieve a certain level of information erasure. This context-dependent Landauer's principle will be characterised by the equivalence class of unitary operators that achieve our task. Of course,  this first requires a re-examination of what exactly we mean by information erasure.

\subsection{Information erasure: pure state preparation and entropy reduction}
In this article, we take  information erasure to be synonymous with pure state preparation; just as in classical mechanics erasure (in the Landauer sense) involves the many-to-one mapping on the information bearing degrees of freedom, then in quantum mechanics this translates naturally as the irreversible process of preparing the object in a pure state. Probabilistic information erasure, then, refers to the case where the probability of preparing the object in the desired pure state is lower than unity.   Although erasing the information of an object as presently defined leads to a reduction of its entropy, the two processes are not quantitatively the same. If we wish to maximise the largest eigenvalue in the object's probability spectrum, thereby maximising the probability of preparing it in a given pure state,  in general we need not minimise its entropy to do so; the only cases where maximising the probability of information erasure leads to minimising the entropy  are when the object has a two-dimensional Hilbert space, or where we are able to fully purify the object and thereby take its entropy to zero. In general, then, a given probability of information erasure is compatible with many different values of entropy reduction. By choosing the smallest entropy reduction, one would expect that we may minimise the consequent heat dissipation, as per \eq{Landauer-equality}.  Consequently, our desired task can be stated as the  minimisation of heat dissipation given probabilistic information erasure -- that is to say, of minimising the amount of energy transferred to the thermal reservoir as heat if we require that the  probability of preparing the object  in a specific pure state $\ket{\varphi_1}$ be no smaller than $p_{\varphi_1}^\mathrm{max}-\delta$. Here $p_{\varphi_1}^\mathrm{max}$ is the maximum probability of information erasure that is permissible by the physical context, and $\delta\geqslant 0$ the error. We will refer to the equivalence class of unitary operators that achieve this as $[U_\mathrm{opt}(\delta)]$. If the object also has a non-trivial Hamiltonian, then to further reduce the total work cost of information erasure, conditional on  first minimising the heat dissipation, we may further optimise the unitary operators within the equivalence class $[U_\mathrm{opt}(\delta)]$ so that the state of the object is made to be  passive \cite{Pusz-Passive-States,Energetic-Passive-State}, and with as small an expected energy value as possible. This reduced equivalence class is referred to as $[U_\mathrm{opt}^\mathrm{p}(\delta)]$.  

\subsection{Information erasure and information processing}
Reducing the heat dissipation due to information erasure is important for both classical and quantum information processing devices. As recent studies suggest \citep{ICT-scaling}, heat dissipation is a major limiting factor on the continual growth in the computational density of modern CMOS transistors. Meanwhile for quantum computation in the circuit-based model,  error correction requires a constant supply of ancillary qubits, in pure states, for syndrome measurements. Indeed, the authors in \citep{Quantum-refrigerator} show that in the absence of such a constant supply the number of steps in which the computation can be performed fault tolerantly
will be limited. Given a finite supply of ancillary qubits,  we must constantly purify them during the execution of the algorithm. If the resulting heat dissipation leads to the intensification of thermal noise beyond the threshold for fault tolerance \citep{Aharonov-fault-tolerance}, then the computation will fail.    
 A context-dependent Landauer's principle will thus prove especially important for information processing devices, in both classical and quantum architectures, where the structure of the reservoir Hamiltonian will  usually be fixed. Furthermore, our work may be useful for certain  high-performance, probabilistic (classical) information processing devices, that would operate at or near the quantum regime.    Although the current state of the art in information processing devices dissipates heat orders of magnitude in excess of Landauer's limit, our ever increasing ability to control microscopic devices will mean that achieving such theoretical limits may be possible in the not-too-distant future. Indeed, experiments already exist, both in classical \citep{Landauer-experiment-classical} and quantum \citep{Goold-experimental-Landauer} systems, which have achieved heat dissipation very close to Landauer's limit.           

\subsection{Layout of article}

  In \sect{optimal unitary section} we shall  characterise the equivalence class of unitary operators acting on the composite system of object and reservoir, as a result of which the object  undergoes probabilistic information erasure and, given this,  the  reservoir gains the minimal quantity of heat. If the object also has a non-trivial Hamiltonian, the unitary operators can be further optimised so as to reduce the energy gained by the object. Here, we operate within  Landauer's framework --   the object and reservoir are initially uncorrelated and  the composite system evolves unitarily.    We demonstrate, using a sequential swap algorithm introduced in \sect{trade-off relation}, the tradeoff  between probability of information erasure and minimal heat dissipation; an increase in probability of preparing the object in a defined pure state is accompanied by an increase in the minimal heat that must be dissipated to the thermal reservoir.    In \sect{examples} we apply the general results to the case of  erasing a    maximally mixed qubit with the greatest allowed probability of success. Two reservoir classes will be considered: (i) a  $d$-dimensional ladder system, where the energy gap between consecutive eigenstates is uniformly $\omega$; and (ii)  a spin chain with nearest-neighbour interactions, that is under a local magnetic field gradient.  For both models, we shall also inquire into the effect of energy conserving,  pure dephasing channels on the erasure process. In \sect{qudit erasure calculations}, we determine the minimum quantity of heat that must be dissipated given full information erasure of a general qudit prepared in a maximally mixed state, in the limit of utilising an infinite-dimensional ladder system, which is a harmonic oscillator.  In \sect{resetting with correlations} we shall address how information erasure can be achieved at a lower heat cost than  Landauer's limit, by operating outside of Landauer's framework, but in such a way that  terms like heat and temperature would continue to have referents in the mathematical description. In \app{Majorisation theory section} we provide a brief overview of certain key results from majorisation theory that will be used throughout the article. In \app{equivalence class of unitary operators} we explain what an equivalence class of unitary operators constitutes. Finally, in \app{technical proofs} we provide proofs for the main results.

\section{Information erasure within Landauer's framework}\label{optimal unitary section}

\subsection{The setup} 

\begin{figure}[!htb]
\centering
\includegraphics[width=10 cm]{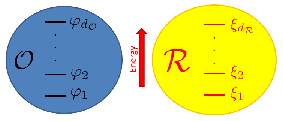}
\caption{The  object $\oo$ with Hilbert space $\h_\oo \simeq\co^{d_\oo}$ and thermal reservoir $\rr$ with Hilbert space $\h_\rr \simeq\co^{d_\rr}$. The eigenbasis of the reservoir Hamiltonian $H_\rr$ is $\{\ket{\xi_m}\}_m$, with the vector numbering being in order of increasing energy.  The eigenbasis with respect to which the object is initially diagonal is $\{\ket{\varphi_n}\}_n$.  }
\label{setup}
\end{figure}

   We consider a system composed of an object, $\oo$, with Hilbert space $\h_\oo \simeq \co^{d_\oo}$ and reservoir, $\rr$, with Hilbert space $\h_\rr \simeq \co^{d_\rr}$. Let the Hamiltonian of the reservoir be the self-adjoint operator $H_\rr=\sum_{m=1}^{d_\rr} \lambda_m^{\uparrow} \pr{\xi_m}$, where  $\bm{\lambda}^{\uparrow}:=\{\lambda_m^{\uparrow}\}_m$ is a non-decreasing vector of energy eigenvalues. This means that $\lambda_i^\uparrow \leqslant \lambda_j^\uparrow$  for any $i<j$.  Similarly, the object Hamiltonian is denoted $H_\oo$.  The compound  system is initially in the uncorrelated state $\rho= \rho_\oo \otimes \rho_\rr(\beta)$, where $\rho_\oo :=\sum_{l=1}^{d_\oo} o_l^{\downarrow} \pr{\varphi_l}$ is the initial state of the object, such that  $\bm{o}^{\downarrow}:=\{o_l^{\downarrow}\}_l$ is a non-increasing vector of probabilities. This means that $o_i^\downarrow \geqslant o_j^\downarrow$ for any $i <j$.  Additionally, the reservoir is initially in the Gibbs state $\rho_\rr(\beta):=e^{-\beta H_\rr}/\tr[e^{-\beta H_\rr}]$ at inverse temperature  $\beta:=(k_B T)^{-1} \in (0,\infty)$. \fig{setup} represents the setup diagrammatically.  Because of the ordering on the energy eigenvalues, we may represent this state  as   $\rho_\rr(\beta) :=\sum_{m=1}^{d_\rr} r_m ^{\downarrow}\pr{\xi_m}$,  such that    $\bm{r}^{\downarrow}:=\{r_m^{\downarrow}\}_m$ is a non-increasing  vector of probabilities.   For simplicity, we write the initial state $\rho$ in the equivalent form
\begin{equation}\label{composite system state representation initial}
\rho=\sum_{l=1}^{d_\oo}\sum_{m=1}^{d_\rr} o_l^{\downarrow}r_m ^{\downarrow}\pr{\varphi_l} \otimes \pr{\xi_m}\equiv \sum_{n=1}^{d_\oo d_\rr} p_n ^{\downarrow}\pr{\psi_n}, \end{equation}  
where the non-increasing vector $\bm{p}^{\downarrow}:=\{ p_n^{\downarrow}\}_n$ is the ordered permutation of  $\{o_l ^{\downarrow}r_m^{\downarrow}\}_{l,m}$, and $\{\ket{\psi_n} \in \h_\oo\otimes \h_\rr\}_n$ the associated permutation of  $\{\ket{\varphi_l} \otimes \ket{\xi_m}\}_{l,m}$. We note that this state representation  is unique if and only if there are no degeneracies in the probability distribution $\bm{p}^{\downarrow}$. 
We assume that the total system is thermally isolated, so that the process of information erasure will be characterised by a unitary operator   $U$. The state of the system after the process is complete is therefore  
\begin{equation}
\rho':=U \rho U^\dagger=\sum_{n=1}^{d_\oo d_\rr}  p_n ^{\downarrow}U\pr{\psi_n}U^\dagger.
\end{equation}
The marginal states of $\rho'$  are  $\rho_\oo':= \tr_\rr[\rho']$ and $\rho_\rr':= \tr_\oo[ \rho']$, where $\tr_A[\cdot]$ represents the partial trace, of a composite system $A+B$, over the system $A$.  

As the pure state we wish to prepare the object in is arbitrary up to local unitary operations, for simplicity we choose this to be  $\ket{\varphi_1}$; this is the eigenstate of $\rho_\oo$ with the largest eigenvalue, i.e., $o_1^\downarrow$. The probability of preparing  $\rho_\oo'$ in the state $\ket{\varphi_1}$ is defined as 
\begin{align}
p(\varphi_1| \rho_\oo'):= \<\varphi_1|\rho_\oo'|\varphi_1\> &= \sum_{n=1}^{d_\oo d_\rr}  p_n ^{\downarrow} \<\psi_n|U^\dagger(\pr{\varphi_1}\otimes \one_\rr)U|\psi_n\>, \nonumber \\ & = \sum_{n=1}^{d_\oo d_\rr}  p_n^{\downarrow} g_n(U) \equiv\bm{p}^{\downarrow}\cdot\bm{g(U)},
\end{align} 
where $\bm{g(U)}$ is a vector of  positive numbers $g_n(U):= \<\psi_n|U^\dagger(\pr{\varphi_1}\otimes \one_\rr)U|\psi_n\>$ such that $\sum_n g_n(U)=d_\rr$. In general, we wish to achieve $p(\varphi_1| \rho_\oo') \geqslant p^{\max}_{\varphi_1} - \delta$, where $p_{\varphi_1}^{\max}$ is the maximum probability of information erasure permissible by the physical context, and $\delta \in [0, p^{\max}_{\varphi_1}- o_1^{\downarrow}]$ is the error. As we want the  process  to produce a larger $p(\varphi_1|\rho_\oo')$ than $o_1^\downarrow$, this will lead to a decrease in the von Neumann entropy of $\oo$. The von Neumann entropy of a state $\rho$ is  $S(\rho):= -\tr[\rho \log(\rho)]$.   We define the reduction in  entropy of $\oo$ as $\Delta S:= S(\rho_\oo)-S(\rho'_\oo)$. 

The process is also assumed to be cyclic, meaning that the total Hamiltonian at the start of the process is identical with that at the end. As such,      the total  average energy consumption of the erasure protocol will be 
\begin{align}\label{energy change of battery}
\Delta E := \tr[(H_\oo+ H_\rr)(\rho'- \rho)]&= \tr[H_\oo(\rho_\oo'- \rho_\oo)]+ \tr[H_\rr(\rho_\rr'- \rho_\rr(\beta)], \nonumber \\
&= \Delta W + \Delta Q.
\end{align}
A positive $\Delta E$ implies that the process requires energy from an external work storage device. Conversely, a negative $\Delta E$ implies that the process produces energy that can, in turn, be stored in the work storage device. Here, $\Delta W$ is the energy change in the object, which we call work done on the object, and $\Delta Q$ the energy change in the reservoir, or the heat dissipated to the reservoir.  As shown in \citep{Esposito-Entropy-Correlation,Reeb-Wolf-Landauer}, these terms can also be written as
\begin{align}
\beta\Delta W &=S(\rho_\oo' \| \rho_\oo(\beta))- S(\rho_\oo\| \rho_\oo(\beta))- \Delta S, \label{Work term}
\end{align}
\begin{align}
\beta \Delta Q &= \Delta S + I(\oo:\rr)_{\rho'}+ S(\rho_\rr' \| \rho_\rr(\beta)), \label{Heat term} \end{align}
 where $S(\rho \| \sigma):= \tr[\rho(\log(\rho) - \log(\sigma))]$ is the entropy of  $\rho$ relative to  $\sigma$, and $I(A:B)_\rho:= S(\rho_A)+S(\rho_B)-S(\rho)$ is the mutual information of a state $\rho$ of a bipartite system $A+B$.    As we are only interested in cases where $\Delta S$ is positive, we can infer from the non-negativity of the relative entropy and mutual information that $\Delta Q$ is always positive for information erasure, even though $\Delta W$ may be negative. 

We wish to make the physical interpretation that $\Delta Q$ is energy that is irreversibly lost during the information erasure process, and is  hence qualitatively different in nature from $\Delta W$. For this to be true, it must be impossible to extract work from the reservoir, after the process is complete, by means of a cyclic unitary process involving  the reservoir alone. This is satisfied  if $\rho_\rr'$ is passive, i.e., $\rho_\rr'= \sum_m r_m'^{\downarrow}\pr{\xi_m}$; that is to say, if $\rho_\rr'$ is diagonal in the Hamiltonian eigenbasis, and its eigenvalues are non-increasing with respect to energy. If $\rho_\rr'$ is not passive, as shown by \citep{Allahverdyan-maximal-work} it is possible to extract a maximum amount of work, given as
\begin{equation}
\Delta W^{\max}:= \tr[H_\rr (\rho_\rr' - \rho_\rr^\mathrm{passive})], 
\end{equation}
where $\rho_\rr^\mathrm{passive}$ has the same spectrum as $\rho_\rr'$, but is passive.   As will be shown in the following sections, not only is it possible for $\rho_\rr'$ to be passive, but this is always satisfied in the case of minimal heat dissipation. However, if the dimension of $\h_\rr$ is at least three, and we have access to $N$ copies of $\rho_\rr'$, it may be possible, for a sufficiently large $N$, to have the compound state $\rho_\rr'^{\otimes N}$ be non-passive.   This is called \emph{activation}. Consequently, by keeping the reservoir systems after their utility in the erasure protocol, and then acting globally on this collection, we may be able to retrieve some energy.  The only passive state which cannot be activated, no matter how many copies we have access to, is the Gibbs state \citep{Pusz-Passive-States}. However,  $\rho_\rr'$ will not in general be in a Gibbs state.  To ensure that $\Delta Q$ is truly lost, irrespective of what reservoir is used, we must impose an additional structure. The simplest method  is to impose the condition that the reservoir system is irrevocably lost after the process is complete. For example, if the reservoir system $\rr$ is randomly chosen from an infinite collection of identical systems, but we do not know which particular system was used, then the probability of picking this system again at random, after the erasure protocol, will be vanishingly small.

\subsection{Maximising the probability of information erasure}\label{maximising probability of information erasure}
 In \app{proof maximising probability of information erasure} we prove that the maximum probability of information erasure  is 
\begin{equation}
p_{\varphi_1}^{\max}:= \sum_{m=1}^{d_\rr} p_m^{\downarrow},
\end{equation} and the equivalence class of unitary operators that achieve this, denoted $[U^g_\mathrm{maj}]$,   is characterised by the rule
\begin{equation}
\text{for all } m\in \{1, \dots, d_\rr\} \ , \ U_\mathrm{maj}^g\ket{\psi_m}=  \ket{\varphi_1} \otimes\ket{ \xi_m'},\label{maximising probability rule equation}
\end{equation}
where $\{\ket{\xi_m'}\}_m$ is an arbitrary orthonormal basis in $\h_\rr$.
To see what we mean by an equivalence class of unitary operators, refer to \app{equivalence class of unitary operators}.  In other words, to maximise the probability of information erasure the unitary operator must take the $d_\rr$ vectors $\ket{\psi_m}$, that have the largest probabilities associated with them in the spectral decomposition of $\rho$, to the product vectors $\ket{\varphi_1}\otimes \ket{\xi_m'}$.  
Similar results, leading to the conclusion  that $p_{\varphi_1}^{\max}$ in general cannot be brought to unity, have been reported in \citep{Reeb-Wolf-Landauer, Quantum-resource-cooling-Viola, Dynamical-Cooling,No-go-cooling-Brumer}.

A necessary and sufficient condition for $p^{\max}_{\varphi_1}$ to be greater than the largest eigenvalue of the object's initial state, i.e,  $p(\varphi_1|\rho_\oo):=o_1^{\downarrow}$, is that $o_2^{\downarrow} r_1^{\downarrow}$ be greater than $o_1^{\downarrow} r_{d_\rr}^{\downarrow}$. If this were not the case, the $d_\rr$ largest probabilities $p_m^\downarrow$ would be the set $\{o_1^\downarrow r_m^\downarrow\}_m$. Recall that the maximum probability of information erasure is given by  summing over this set, which   gives $o_1^\downarrow$. That is to say, $p^{\max}_{\varphi_1}:= \sum_{m=1}^{d_\rr} p_m^\downarrow \equiv \sum_{m=1}^{d_\rr} o_1^{\downarrow} r_m^{\downarrow} = o_1^{\downarrow}$. This implies that for a non-trivial  erasure process, whereby the probability of preparing the object in the state $\ket{\varphi_1}$ is increased, we require that
\begin{align}
\frac{o_1^{\downarrow}}{o_2^{\downarrow}} < \frac{r_1^{\downarrow}}{r_{d_\rr}^{\downarrow}}=e^{\beta(\lambda_{d_\rr}^{\uparrow}-\lambda_1^{\uparrow})},
\end{align}
where the equality is a consequence of $r_m^\downarrow:= e^{-\beta \lambda_m^\uparrow}/\tr[e^{-\beta H_\rr}]$.  Similar arguments were made in \citep{Reeb-Wolf-Landauer}, although there the focus was on providing a bound on the smallest eigenvalue of $\rho_\oo'$ that could be obtained. 

\subsection{Minimising the heat dissipation}\label{Minimising the heat dissipation}

As the initial state of the reservoir is fixed, the heat dissipation is minimised by lowering the expected energy of the post-transformation marginal state of the reservoir, $\tr[H_\rr \rho_\rr']$. In \app{proof Minimising the heat dissipation} we prove that $\Delta Q$ is minimised by the equivalence class of unitary operators $[U_\mathrm{maj}^f]$ characterised by the rule
\begin{equation}
\text{for all } m \in \{1, \dots,d_\rr\}  \text{ and } n \in\{(m-1)d_\oo+1,\dots ,m d_\oo\} \ , \ U_\mathrm{maj}^f\ket{\psi_n}= \ket{\varphi_{l}^m} \otimes \ket{\xi_m} ,
\label{minimising heat dissipation rule}\end{equation}
with the set  $\{\ket{\varphi_{l}^m}|l \in \{1, \dots, d_\oo\}\}$ forming an orthonormal basis in $\h_\oo$ for each $m$. A unitary operator from this equivalence class will ensure that $\rho_\rr'$ is passive, and that it majorises any other passive state that could have been prepared. This is done by first maximising the probability of preparing the reservoir in the ground state $\ket{\xi_1}$, by taking the $d_\oo$ vectors $\ket{\psi_n}$, that have the largest probabilities associated with them in the spectral decomposition of $\rho$, to the product vectors $\ket{\varphi_{l}^1} \otimes \ket{\xi_1}$. After this, the probability of preparing the reservoir in the next energy state $\ket{\xi_2}$ is maximised in a similar fashion, and so on for all other energy eigenstates. 

\subsection{Minimal heat dissipation conditional on maximising the probability of information erasure }\label{maximal-probability-optimal-unitary}

\begin{figure}[!htb]
\centering
\includegraphics[width=10 cm]{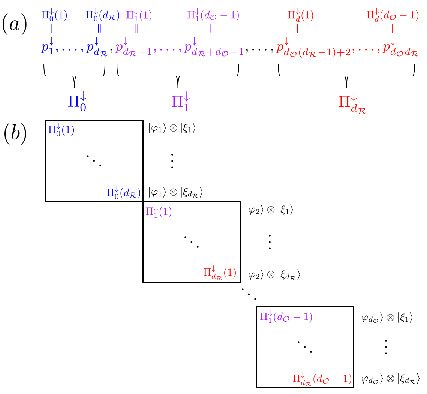}
\caption{(Colour online)  (a) The partitioning of  $\bm{p^{\downarrow}}$, the  vector of eigenvalues of $\rho$ arranged in a non-increasing order, into the vectors $\Pi_0^{\downarrow}$ and $\Pi^{\downarrow}_{m\geqslant 1}$.  (b) The density operator $\rho':=U_\mathrm{opt}^{\mathrm{p}}(0)\rho {U_\mathrm{opt}^{\mathrm{p}}(0)}^\dagger$, in matrix representation, where $U_\mathrm{opt}^{\mathrm{p}}(0)$ is the optimal unitary operator for passive, maximally probable information erasure. The post-transformation marginal state of the object, $\rho_\oo'$, is  passive. It is also the least energetic passive state that is possible to prepare, given the constraints: (i) $p(\varphi_1|\rho_\oo')=p_{\varphi_1}^{\max}$; and (ii) $\Delta Q$ is minimal given (i).   }
\label{optimal unitary schematic}
\end{figure}

If we compare the rule that maximises the probability of information erasure, given by \eq{maximising probability rule equation}, and the rule that minimises the heat dissipation, given by \eq{minimising heat dissipation rule}, we notice that they  are incompatible. As such, no unitary operator simultaneously exists in both equivalence classes:    $[U_{\mathrm{maj}}^g] \cap\ [U_{\mathrm{maj}}^f]=\{\emptyset\}$. The two tasks are in some sense complementary, and there will be a tradeoff between them. Here, we  shall prioritise; a unitary operator will be chosen such that it maximises the probability of information erasure and, given this constraint,  minimises the heat dissipation.
In other words, we find the equivalence class of unitary operators $[U_\mathrm{opt}(0)]\subset[U^g_\mathrm{maj}]$ that minimise $\Delta Q$. The zero in braces indicates that the error in probability of information erasure, $\delta$, is zero.    To this end we first  divide the vector of probabilities $\bm{p ^{\downarrow}}$ to form the non-increasing vector of cardinality $d_\rr$, denoted $\Pi_0^{\downarrow}$, and the non-increasing vectors of cardinality $d_\oo-1$, denoted $\{\Pi_m^{\downarrow}|m \in \{1,\dots, d_\rr\}\}$,  defined as
\begin{align}
\Pi_0^{\downarrow}&:= \{ p_m^{\downarrow}| m \in\{1,\dots,d_\rr\}\}, \nonumber \\
 \Pi^{\downarrow}_{m\geqslant 1}&:= \{ p^{\downarrow}_{d_\rr + (m-1)( d_\oo-1) +l} | l \in \{1, \dots, d_\oo-1\}\}.
\end{align}
 We refer to the $m\ts{th}$   element of $\Pi_0^{\downarrow}$ as $\Pi_0^{\downarrow}(m)$, and the $l\ts{th}$ element of  $\Pi^{\downarrow}_{m\geqslant 1}$ as $\Pi^{\downarrow}_{m\geqslant 1}(l)$. 

In \app{proof maximal-probability-optimal-unitary} we prove that the equivalence class of  unitary operators that maximise the probability of information erasure and, given this constraint, minimise the heat dissipation, is  characterised by the rules 
\begin{align}
U_\mathrm{opt}(0):\begin{cases}\ket{\psi_n} \mapsto \ket{\varphi_1} \otimes \ket{\xi_m} & \text{if }  p(\psi_n|\rho) = \Pi^{\downarrow}_0(m), \\
\ket{\psi_n} \mapsto \ket{\varphi_l^m}\otimes\ket{\xi_{m}} & \text{if }  p(\psi_n|\rho)= \Pi^{\downarrow}_m(l)  \text{ and }  m\geqslant 1,\\
\end{cases}\label{optimal U equation}
\end{align}
where, for all $m$, each member of the orthonormal set $\{\ket{\varphi_l^m}\}_l$ is orthogonal to $\ket{\varphi_1}$.

 Effectively, the first line of \eq{optimal U equation} conforms with \eq{maximising probability rule equation} and hence maximises the probability of information erasure. The orthonormal vectors $\{\ket{\xi_m'}\}_m$ are chosen to be the eigenvectors of the reservoir Hamiltonian, however, in order to minimise the contribution to heat from this line. The second line is an altered version of \eq{minimising heat dissipation rule}, thereby minimising the heat dissipation given the constraint posed by the first line. We now make the following observations: 
\begin{enumerate}[(a)]
\item
If we choose  $\ket{\varphi_l^m}=\ket{\varphi_{l+1}}$ for all $m$, and such that $\{\ket{\varphi_l}\}_l$ are the eigenvectors of   the object Hamiltonian $H_\oo$ in increasing order of energy, then $U_\mathrm{opt}(0)$ would also ensure that erasure  to the ground state $\ket{\varphi_1}$ would  be done in such a way that $p(\varphi_i|\rho_\oo')\geqslant p(\varphi_j|\rho_\oo') $ for all $i<j$; the object is brought to  a passive state, although this state will in general not be thermal \citep{Pusz-Passive-States}. We refer to this as \emph{passive information erasure}, and the resultant equivalence class of  unitary   operators as $[U_\mathrm{opt}^{\mathrm{p}}(0)]\subset [U_\mathrm{opt}(0)]$.  These unitary operators will result in the smallest possible $\Delta E$, conditional on first maximising the probability of information erasure, and then minimising the heat dissipation; that is to say, $[U_\mathrm{opt}^{\mathrm{p}}(0)]$ minimises $\Delta W$ for all unitary operators in the equivalence class $[U_\mathrm{opt}(0)]$.   \fig{optimal unitary schematic} shows the matrix representation of $\rho'= U_\mathrm{opt}^{\mathrm{p}}(0) \rho {U_\mathrm{opt}^{\mathrm{p}}(0)}^\dagger$. 
\item Since the desired task is the maximisation of $p(\varphi_1|\rho_\oo')$, we need not in general maximise $\Delta S$ because this will lead to a greater amount of heat dissipation than necessary, as per \eq{Heat term}. The only cases where maximisation of $p(\varphi_1|\rho_\oo')$ necessarily leads to the maximisation of $\Delta S$  are when: (i) $p_{\varphi_1}^{\max}=1$;  and (ii) where $\h_\oo \simeq  \co^2$. In case (i) the entropy of the object is brought to zero, so $\Delta S$ is trivially maximised. In case (ii), we note that  if $\bm{o_1}^\downarrow \succ \bm{ o_2}^\downarrow$, where $\bm{o_1}^{\downarrow}$ and $\bm{o_2}^{\downarrow}$ are the  spectra of $\rho^1_\oo$ and $\rho_\oo^2$ respectively, then $S(\rho_\oo^1)\leqslant S(\rho_\oo^2)$. If we maximise $p(\varphi_1|\rho_\oo')$ in the case of $\oo$ being a two-level system, this will necessarily  minimise $p(\varphi_2|\rho_\oo')$, which in turn will result in the  spectrum of $\rho_\oo'$ to majorise all possible spectra. Consequently,  $S(\rho_\oo')$ will be minimised, and hence  $\Delta S$ will be maximised. 

However, one can always say that maximising the probability of information erasure requires that we minimise the min-entropy, $S_{\min}$, defined as 
\begin{equation}
S_{\min}(\rho):=\min_i\{-\log(p_i)\},
\end{equation}
where $\{p_i\}_i$ is the  spectrum of $\rho$ \citep{Tomamichel-thesis}. The min-entropy is clearly given by the largest value in the spectrum. To minimise the min-entropy, therefore, we must maximise the largest value in the spectrum; this is  the definition of maximising the probability of information erasure. 
\item The only instance where  $\h_\oo\simeq \h_\rr \simeq \co^d$, and $U_\mathrm{opt}^{\mathrm{p}}(0)$ for passive, maximally probable information erasure is a swap operation, is when $d=2$. For larger dimensions, this is no longer the case. 

\item It is evident that $\rho_\rr'$ is diagonal with respect to the eigenbasis of $H_\rr$, and that the spectrum of $\rho_\rr'$  is non-increasing with respect to the eigenvalues of $H_\rr$. In other words, $\rho_\rr'$ is a passive state. However,  its spectrum is majorised by that of $\rho_\rr(\beta)$. As such, by \corref{corollary-monotone-lattice-superadditive}, $\Delta Q \geqslant 0$. This conforms with Landauer's principle that information erasure must dissipate heat. 

\end{enumerate}

\subsection{The tradeoff between probability of information erasure and minimal heat dissipation}\label{trade-off relation}

 We would now like to relax the condition of maximising the probability of information erasure, and allow the error $\delta$ to take non-zero values. The question we would now like to ask is: how will the minimal achievable $\Delta Q$ be affected by varying $\delta$, and how may we then characterise the equivalence class of  optimal unitary operators $[ U_\mathrm{opt}^{\mathrm{p}}(\delta)]$ ?  The answer for the extremal cases is trivial; we have already addressed the case of $\delta=0$ in \sect{maximal-probability-optimal-unitary}, and when $\delta=p^{\max}_{\varphi_1} -  o_1^{\downarrow}$, then $[ U_\mathrm{opt}^{\mathrm{p}}(\delta)]=\one$ and $\Delta Q=0$.  In \app{proof trade-off relation} we prove that the  algorithm  of sequential swaps,  shown in \fig{sequential swap algorithm}, will result in a non-increasing sequence of errors, $\bm{\delta^\downarrow}:=\{\delta_j^\downarrow\}_j$, commensurate with a non-decreasing sequence of heat, $\bm{\Delta Q ^\uparrow}:= \{\Delta Q_j^\uparrow\}_j$.  For each error $\delta_j^\downarrow$, the associated value of heat $\Delta Q_j^\uparrow$ will be minimal. Furthermore, the marginal state of the object, $\rho_\oo'$, will always be passive. Each swap operation acts on a  subspace spanned by $\{\ket{\varphi_i},\ket{\varphi_j}\}\otimes \{\ket{\xi_k},\ket{\xi_l}\}$.  As the state is initially diagonal with respect to the basis $\{\ket{\varphi_l}\otimes \ket{\xi_m}\}_{l,m}$, and swap operations only permute the probabilities in the state's spectrum, the composite system will always be diagonal with respect to this basis at every stage of the algorithm.  
\begin{figure}[!htb]
\begin{tabular}{|c||c|}\hline
Step (1) & Set $i=2$ and $m=d_\rr$. \\\hline
Step (2) & Sequentially swap  $\ket{\varphi_1}\otimes \ket{\xi_{i}}$  with the vectors $\ket{\varphi_{l}}\otimes \ket{\xi_{m}}$ \\ &  with $l$ running from $d_\oo$ down through to $2$, only if $p_{1,i}< p_{l,m}$. \\\hline
Step (3) & If $m>1$, set $m= m-1$ and go back to Step (2). Else, proceed to Step (4). \\\hline
Step (4) & If $i <d_\rr$, set $i=i+1$, $m=d_\rr$, and go back to Step (2). Else, terminate.  \\\hline
\end{tabular}
\caption{A sequence of swap operations that results in a non-increasing sequence of errors, $\bm{\delta^\downarrow}$, commensurate with a non-decreasing sequence of minimal heat $\bm{\Delta Q^\uparrow}$. At each stage of the algorithm, the probability associated with the vector $\ket{\varphi_i}\otimes \ket{\xi_j}$ is denoted as $p_{i,j}$.}
\label{sequential swap algorithm}
\end{figure}

\begin{figure}[!htb]
\centering
\includegraphics[width=10 cm]{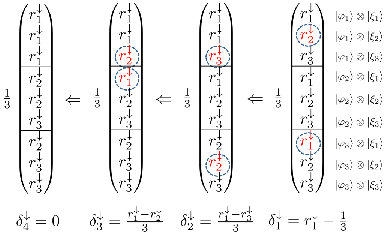}
\caption{(Colour online) The diagonal elements of $\rho':=U^\mathrm{p}_\mathrm{opt}(\delta)\rho {U^\mathrm{p}_\mathrm{opt}(\delta)}^\dagger$, for $\rho = \frac{1}{3} \one_\oo \otimes \rho_\rr(\beta)$, resulting in $p(\varphi_1|\rho_\oo')=p^{\max}_{\varphi_1}-\delta$.  $\Delta Q$ is minimised and $\rho_\oo'$ is  passive with the smallest average energy possible given this constraint.   Here $\h_\oo \simeq \h_\rr \simeq \co^3$, and  $\{\delta_j^{\downarrow}\}_j$ is a non-increasing sequence of errors. The elements inside a dashed circle (red online) are those which must be swapped to move from $\delta_j^{\downarrow}$ to $\delta_{j+1}^{\downarrow}$.   }
\label{error reset}
\end{figure}

\fig{error reset} depicts this process for the case where $\h_\oo\simeq\h_\rr \simeq \co^3$, with $\rho_\oo=\frac{1}{3}\one_\oo$. Here the diagonal entries of the density operator $\rho'$ are shown in each column, with the first column from the right representing the initial state, and the final column representing the case of passive, maximally probable information erasure. The algorithm for reducing error by increasing heat moves from  right to left, as shown by the arrows.  The elements surrounded by dashed circles, and coloured in red, are those which must be swapped to decrease $\delta$, with the resultant diagonal elements of the new state shown to the left.

To allow for a continuous change in $\delta$, we need to generalise the swap operation to an entangling swap. That is to say, for the vectors $\ket{\varphi_1}\otimes \ket{\xi_{i}}$ and $\ket{\varphi_{l}}\otimes\ket{\xi_m}$, and the real number $\gamma \in [0,1]$, we define
\begin{equation}
\mathrm{SW}_\gamma: \begin{cases}\ket{\varphi_1}\otimes \ket{\xi_{i}} \mapsto \sqrt{1-\gamma } \ket{\varphi_1}\otimes \ket{\xi_{i}} +\sqrt{\gamma}\ket{\varphi_{l}}\otimes\ket{\xi_m},  \\
\ket{\varphi_{l}}\otimes\ket{\xi_m} \mapsto  \sqrt{\gamma } \ket{\varphi_1}\otimes \ket{\xi_{i}} -\sqrt{1-\gamma}\ket{\varphi_{l}}\otimes\ket{\xi_m}. \\
\end{cases}
\end{equation} 
Therefore, $\mathrm{SW}_0=\one$ and as $\gamma \to 1$, $\mathrm{SW}_\gamma$ converges to the swap operation. Hence, for any error $\delta \in (\delta_j^{\downarrow}, \delta^{\downarrow}_{j+1})$, the optimal unitary operator $U_\mathrm{opt}^p(\delta)$ would be given by following the algorithm for discrete errors up to $\delta^{\downarrow}_j$, and then replacing the swap operation which would give the error $\delta^{\downarrow}_{j+1}$ with the entangling swap operation defined above, with an appropriate choice of $\gamma$. This will ensure for a continuous decrease in $\delta$ and a continuous increase in $\Delta Q$.

\section{Examples: Erasing a fully mixed qubit with maximal probability of success}\label{examples}
We shall now consider the erasure of a qubit, with Hilbert space $\h_\oo \simeq\co^2$.  We are also  interested in examining the   scenario where no a priori information about the state of the object is known; the probabilities $o_1^\downarrow$ and $o^\downarrow_2$ are both one-half. For simplicity, we make the substitution $d_\rr\equiv d$ for the dimension of the reservoir's Hilbert space. The action of the optimal unitary operator for passive, maximally probable information erasure, would therefore be such that the diagonal elements of $\rho'$, as depicted in Fig.~(\ref{optimal unitary schematic}(b)) and from top to bottom in decreasing order, are the probabilities $\bm{p^\downarrow}\equiv \{\frac{r_1^\downarrow}{2},\frac{r_1^\downarrow}{2},\dots, \frac{r_d^\downarrow}{2}, \frac{r_d^\downarrow}{2}\}$.     We will consider two models for the reservoir: 

\begin{enumerate}[(a)]

\item
\textbf{A  ladder system}.

The ground state of the system has an energy of zero, and for every $m$,
\begin{equation}
\<\xi_{m+1}| H_\rr |\xi_{m+1}\>-\<\xi_m|H_\rr |\xi_m\>=\omega.
\end{equation}
The  $m\ts{th}$  energy of such a system, in increasing order,   is given as $\lambda_m^\uparrow= \omega(m-1)$. The Hamiltonian  has the operator norm $\|H_\mathcal{\rr}\|=\lambda_d^\uparrow=\omega(d-1)$ which grows with $d$. In the limit as $d$ tends to infinity, this system will be a harmonic oscillator of frequency $\omega$, with a spectrum bounded from below by zero, and unbounded from above. 

\item \textbf{A chain of spin-half systems, with nearest-neighbour interactions, that are under a linear magnetic field gradient}

Here, the reservoir has the Hilbert space $\h_\rr=\bigotimes_{k=1}^N \h_k$, with $\h_k\simeq\co^2$ for all $k$. The Hamiltonian is 
\begin{equation}
H_\rr= \sum_{k=1}^N (k\Theta) \sigma_z^k + J\sum_{k=1}^{N-1} \sum_{a \in \{x,y,z\}}\sigma_a^k \otimes \sigma_a^{k+1},
\label{spin-chain Hamiltonian}
\end{equation}
where $\{\sigma_a|a \in \{x,y,z\}\}$ are the Pauli operators. The operator $\sigma_a^k$ acts nontrivially only on Hilbert space $\h_k$.  The parameters  $\Theta \in \re^+$ and $J\in \re^+$ represent, respectively, an effective magnetic field gradient in the $z$-axis and the nearest-neighbour spin-spin coupling strength. This Hamiltonian conserves the total magnetisation, $\sum_k \sigma_z^k$.

\end{enumerate}

 For each reservoir, we wish to determine how much heat is dissipated   in excess of the improved lower bound of Landauer's inequality, determined in \citep{Reeb-Wolf-Landauer},    given as
\begin{equation}
\Delta L:= \Delta Q - \frac{1}{\beta} \left(\Delta S + \frac{2(\Delta S)^2}{\log^2(d-1)+4} \right).\label{Landauer harmonic oscillator offset}
\end{equation} 
We use the simple form of this lower bound, which is not tight.

\subsection{Comparison of reservoirs given unitary evolution}\label{section blah}

\begin{figure}[!htb]
\centering
\subfigure[Ladder system, $\omega=1$ and varying $d$.]{\label{reset-landauer}
\includegraphics[width=8 cm]{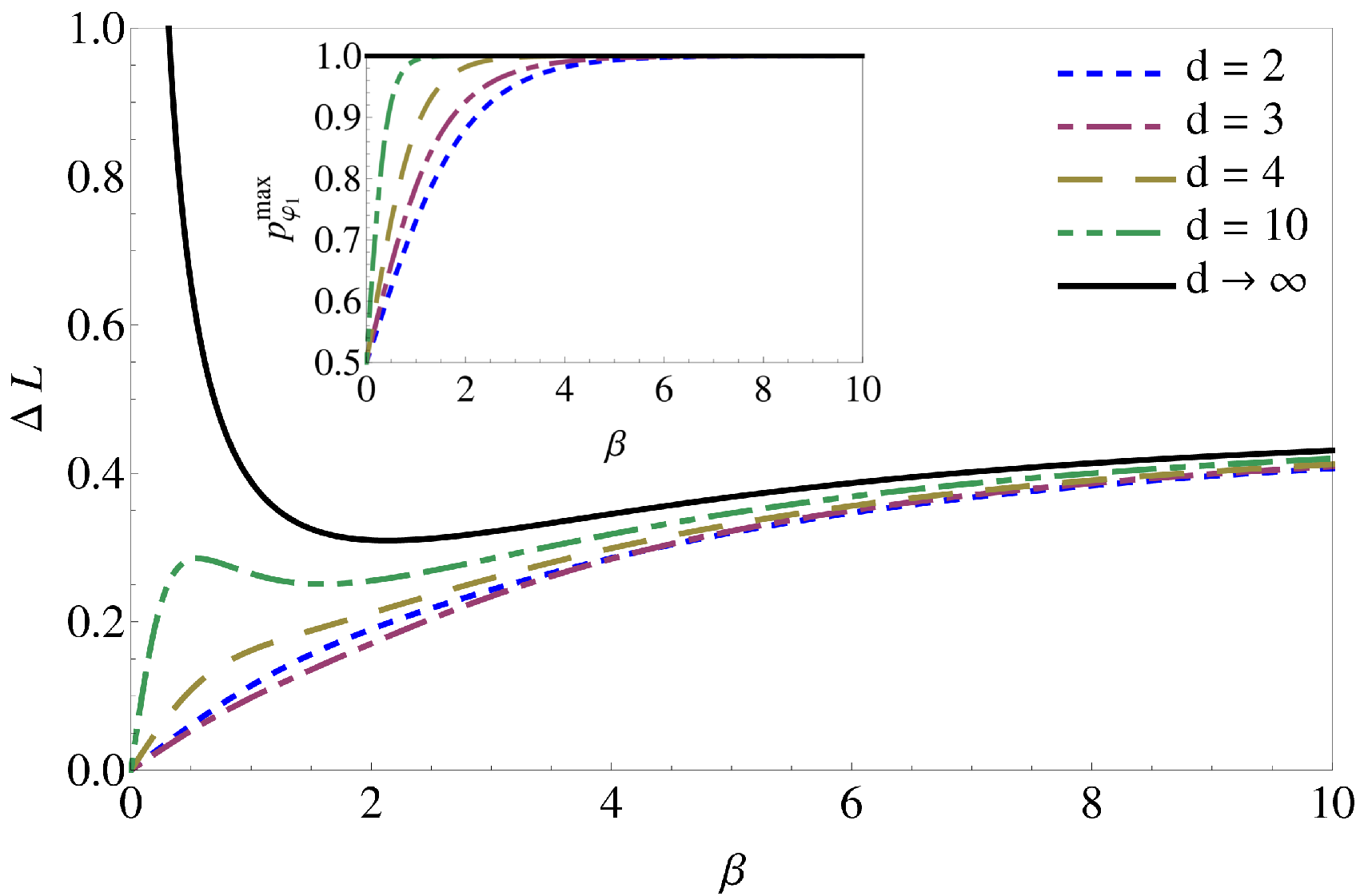}}
\subfigure[Spin chain, $J=\beta=1$,  and varying $\Theta$.]{\label{reset-spin-chain-local-field}
\includegraphics[width=8 cm]{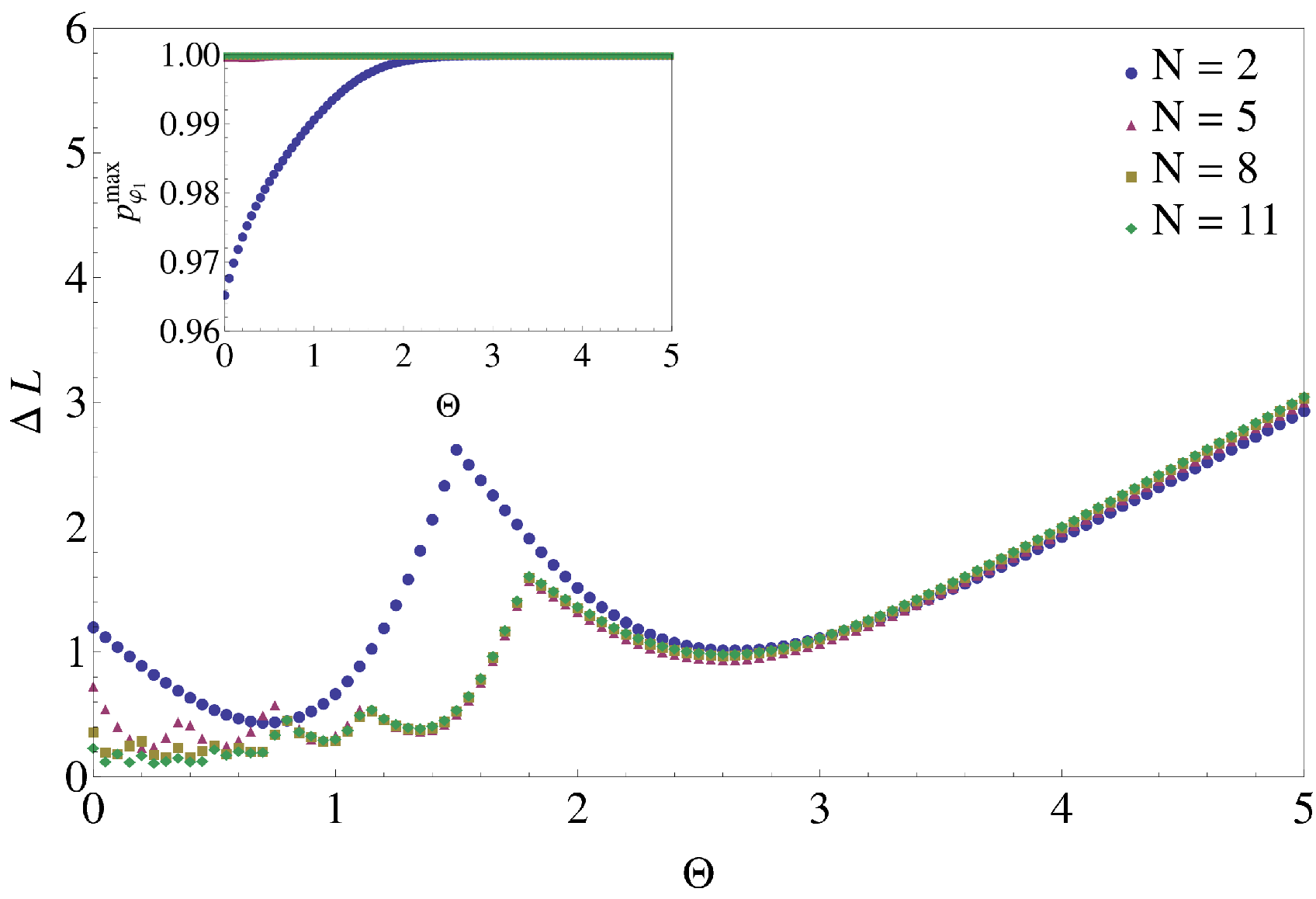}}
\subfigure[Spin chain, $\Theta=\beta=1$,  and varying $J$.]{\label{reset-spin-chain-coupling}
\includegraphics[width=8 cm]{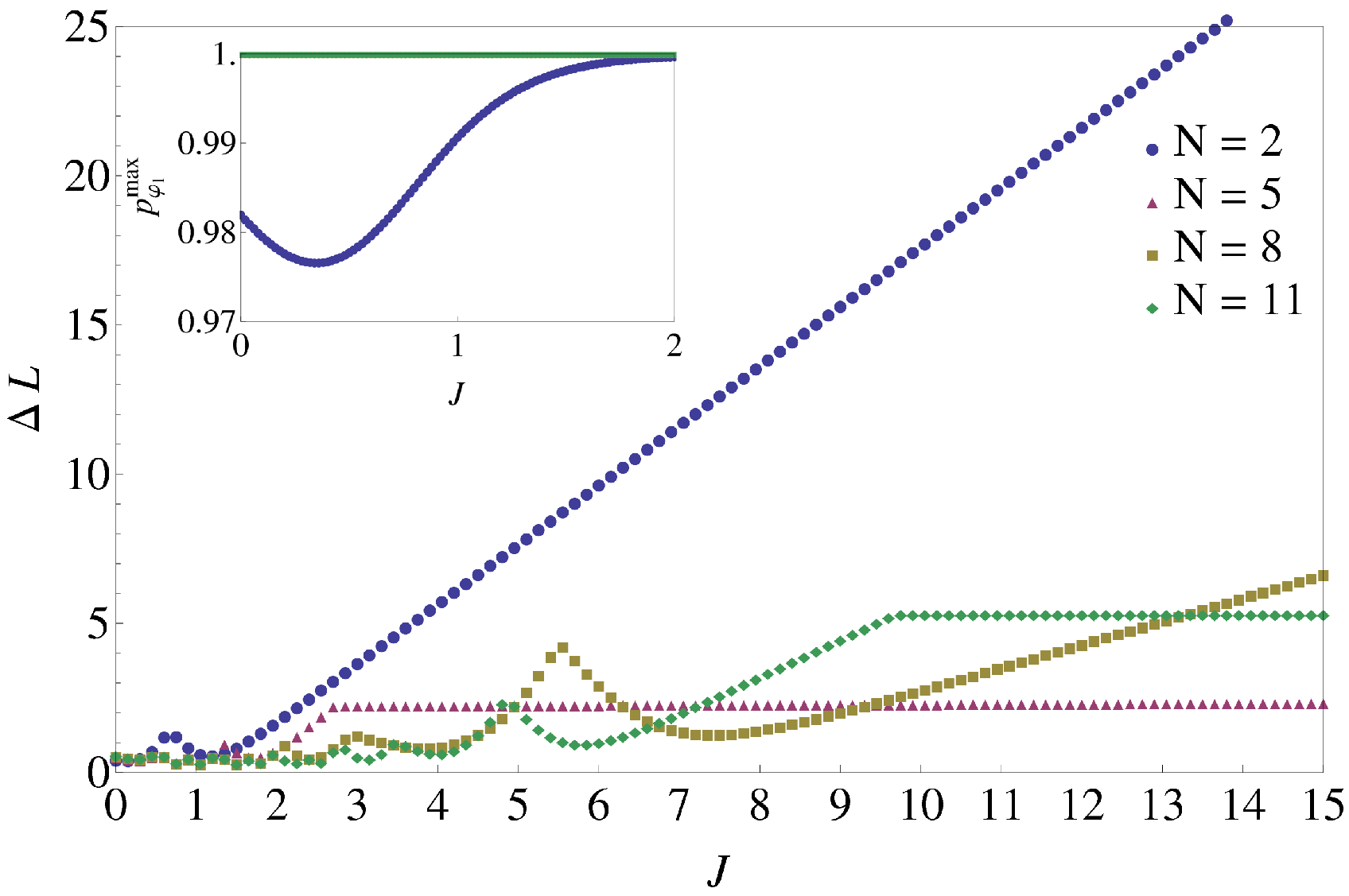}}
\subfigure[Spin chain, $\Theta=J=1$,  and varying $\beta$.]{\label{reset-spin-chain-temperature}
\includegraphics[width=8 cm]{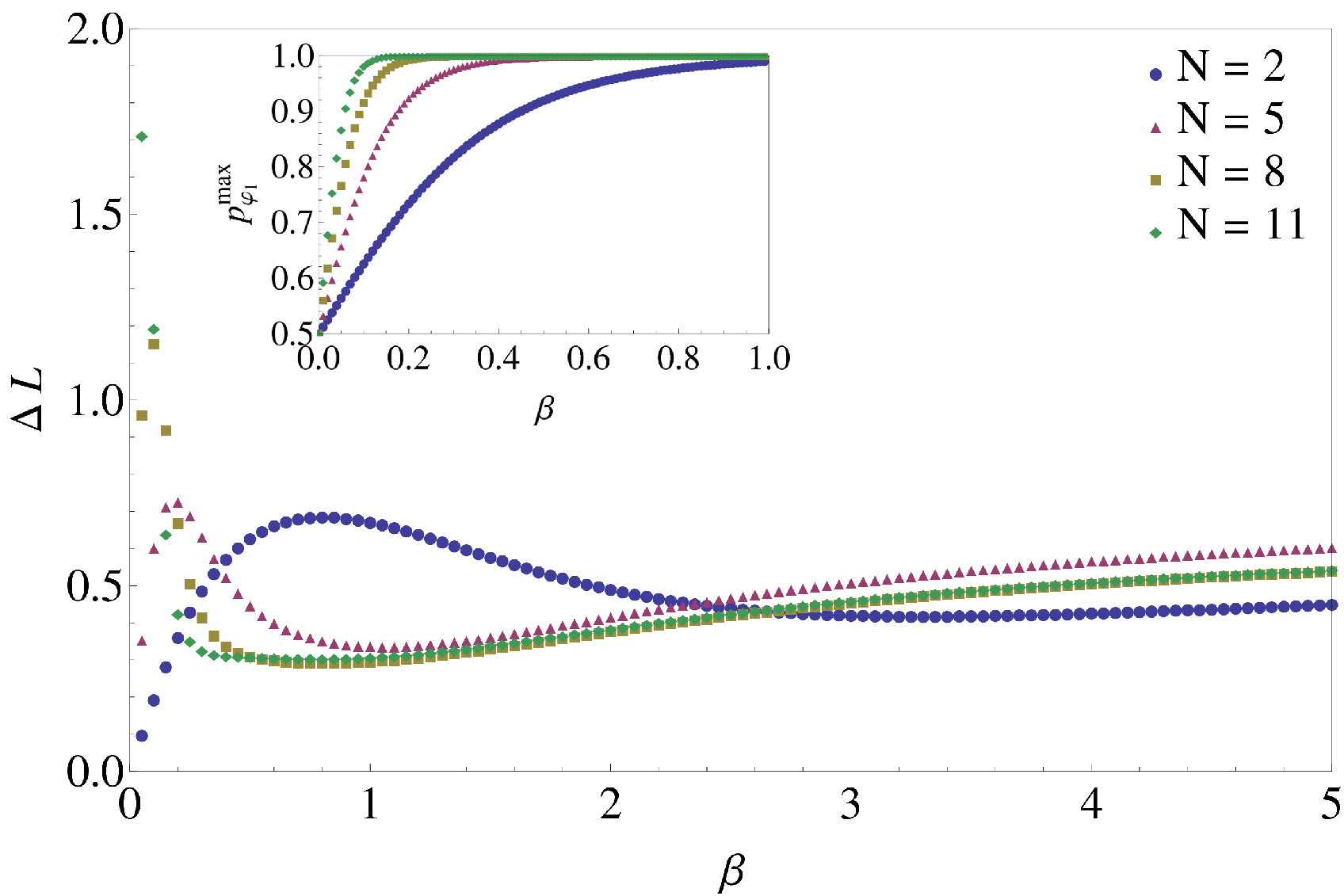}}
\caption{ (Colour online)  dependence of $\Delta L$ and $p_{\varphi_1}^{\max}$ as a function of one parameter. (a) The reservoir is formed by a ladder system with energy spacing $\omega=1$.  (b)-(c) Here the reservoir is formed by a chain of $N$ spins with nearest-neighbour Heisenberg coupling $J$ and linear magnetic field gradient $\Theta$.      }
\label{Landauer offset}
\end{figure}

 \fig{reset-landauer} demonstrates the dependence of $\Delta L$ and $p_{\varphi_1}^\mathrm{max}$ on $\beta$ and $d$, when the reservoir is a ladder system with a fixed frequency $\omega=1$. \fig{reset-spin-chain-local-field}-\fig{reset-spin-chain-temperature} depict the dependence of $\Delta L$ and $p_{\varphi_1}^\mathrm{max}$ on $\Theta$, $J$, and $\beta$ when the reservoir is a spin chain of length $N$. When varying any of these, the other two are left constant at the value of one. We now make the following observations:

\begin{enumerate}[(a)]
\item When the reservoir is a ladder system,  an increase in  $d$ increases  $p_{\varphi_1}^\mathrm{max}$ and also, generally,  $\Delta L$, for all finite temperatures.  In the limit as $\beta$ tends to infinity, $\rho_\rr(\beta)=\pr{\xi_1}$ and $U_\mathrm{opt}^\mathrm{p}(0)$ effects a swap map in the subspace $\{\ket{\varphi_1},\ket{\varphi_2}\}\otimes \{\ket{\xi_1},\ket{\xi_2}\}$. As such, in this limit $p_{\varphi_1}^\mathrm{max}$ and $\Delta Q$ tend to unity and $\omega/2$ respectively.  
\item When the reservoir is a spin chain, as $N$ increases, so does $p_{\varphi_1}^{\max}$. This can be explained by noting that $p_{\varphi_1}^{\max}$ grows with $\lambda_d^\uparrow-\lambda_1^\uparrow$, which is always greater than or equal to $2 \Theta \sum_{k=1}^N k$.
\item For spin chains of even length, in the limit of large $J$,  $\Delta L$ quickly diverges. However, there is some critical value of $J$ for odd-length chains such that an increase in $J$ beyond this drastically reduces the rate at which    $\Delta L$ increases. 
\item The best case scenario is when the reservoir is a long chain, with $\Theta < J,\beta$ and $J \sim \beta$.  For example, for a chain of eleven spins, with $\Theta=0.25$, and $J=\beta=1$, we get $p_{\varphi_1}^{\max}\approx 1$ while $\Delta L \approx 0.12$. Compare this with the case where the reservoir is given by a ladder system of dimension $d=2^{11}$ and $\beta=1$. Here, in order to achieve the same value of $p_{\varphi_1}^{\max}$, realised when $\omega \approx 0.1$, we get $\Delta L \approx 0.29$.  
\end{enumerate} 

\subsection{Comparison of reservoirs under energy-conserving, Markovian dephasing channels}

\begin{figure}[!htb]
\subfigure[Spin chain, $\Theta=J=\beta=1$,  and varying $\Gamma$.]{\label{reset-spin-chain-dephasing}
\includegraphics[width=8 cm]{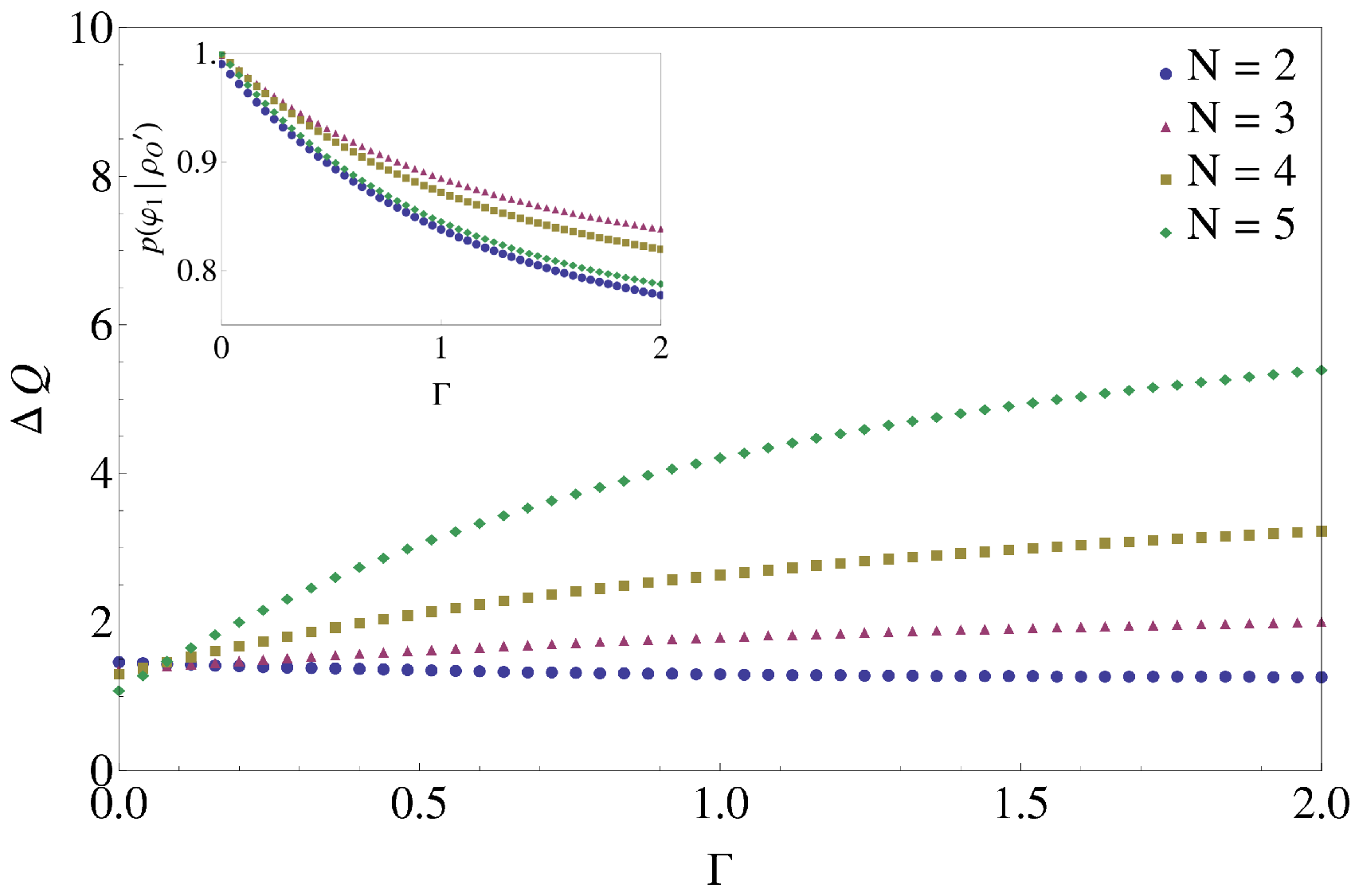}}
\subfigure[Ladder system, $\omega=1$, $\beta=5$, and varying $\Gamma$.]{\label{reset-landauer-dephasing}
 \includegraphics[width=8 cm]{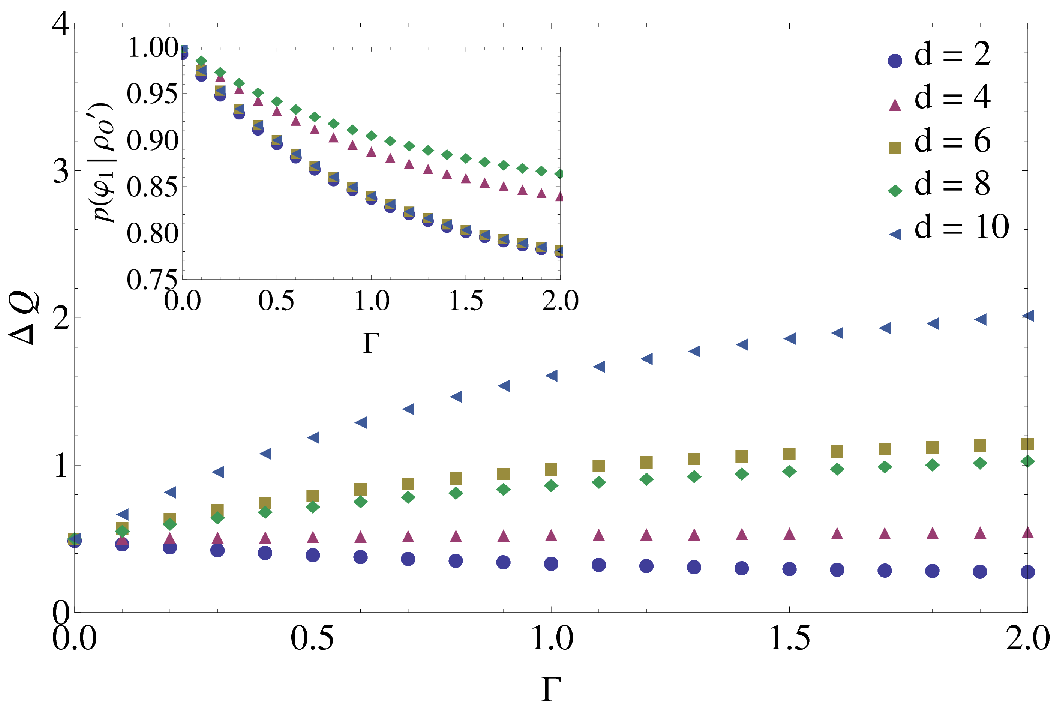}}

\subfigure[Ladder system, $\beta=5$ and $\omega=\Gamma=1$.]{\label{dephasing-heat-dimension}
\includegraphics[width=8 cm]{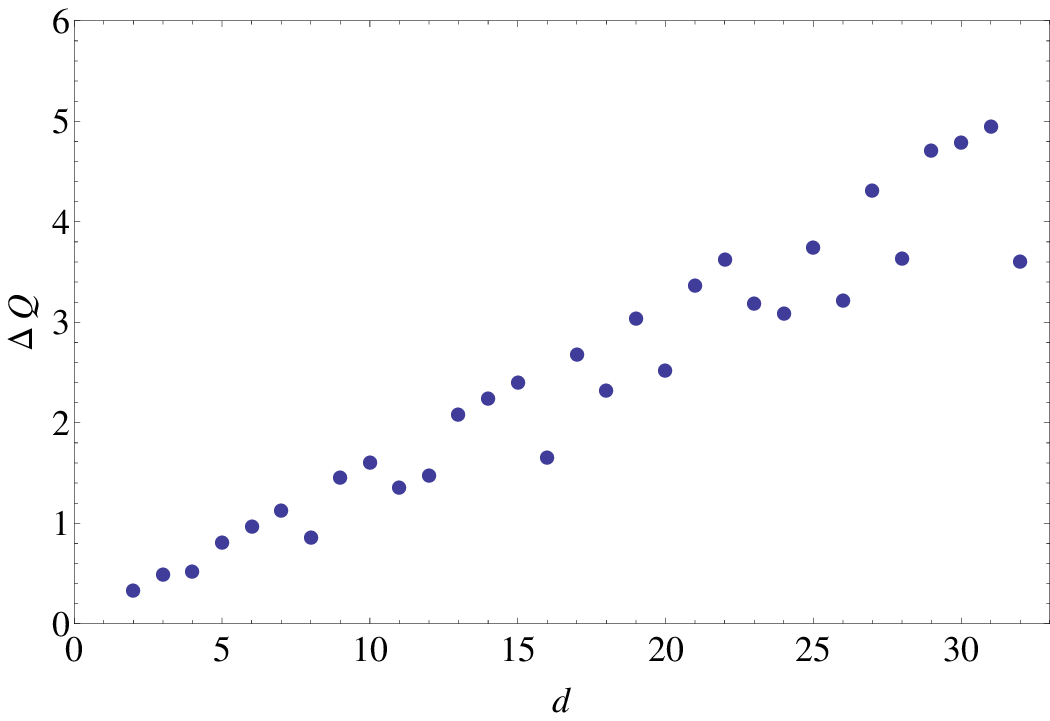}}
\subfigure[Ladder system, $\beta=5$ and $\omega=\Gamma=1$.]{\label{dephasing-probability-dimension}
\includegraphics[width=8 cm]{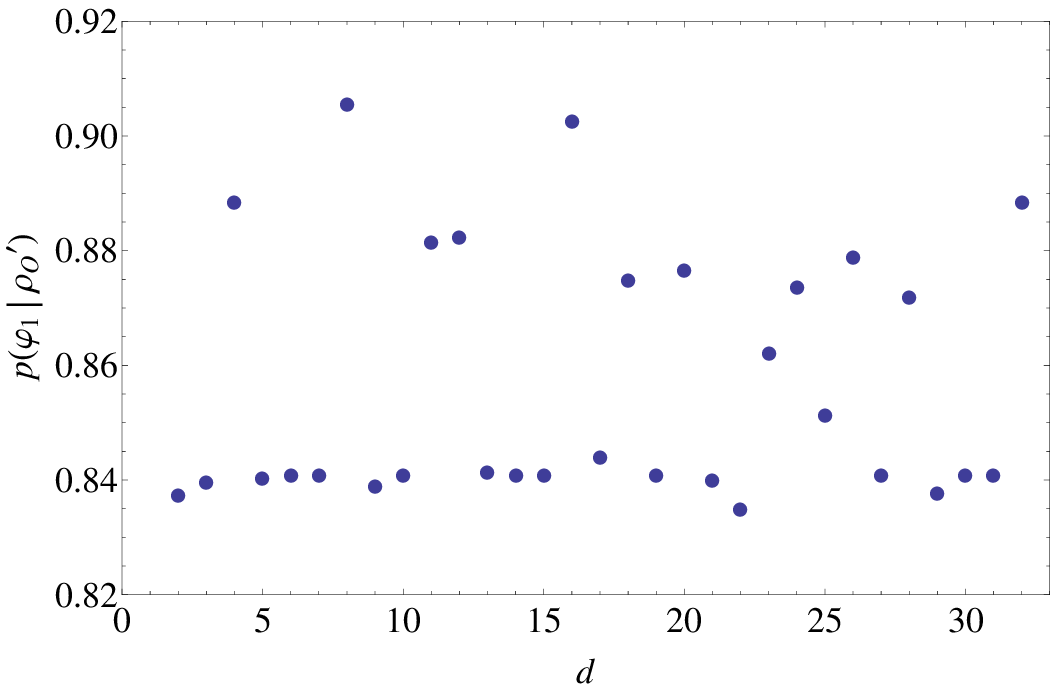}}
\caption{ (Colour online)  (a) and (b)  show the effect of dephasing rate $\Gamma$  on   $ \Delta Q $ and $p(\varphi_1|\rho_\oo')$. The system is evolved for time $\tau=1$.  (a) The reservoir is given by a spin chain of length $N$ and where all the parameters are set to one. (b) The reservoir is given by a ladder system, with energy spacing  $\omega=1$,  and inverse temperature $\beta = 5$.  (c) and (d) show, respectively, the  effect of ladder system dimension $d$ on $\Delta Q$ and $p(\varphi_1|\rho_\oo')$, at a constant value of $\Gamma =1$. It appears that for dimensions $d=2^n$, with $n \in \mathds{N}$ and $n\geqslant 2$,  $\Delta Q$ is smaller  than that of all larger dimension values, while $p(\varphi_1|\rho_\oo')$ take the largest global values. In other words, the ladder system is most robust to energy conserving, Markovian dephasing, when it is dimensionally equivalent to a spin chain.   }
\end{figure}

Before this juncture, we have considered the active element of erasure -- the unitary operator -- as  a bijection between two orthonormal basis sets. To consider this as a bona fide dynamical process we must conceive of the time-ordered  sequence $\{H_{k} |k \in \{1,\dots,N\}\}$, where $H_{k}$ is the Hamiltonian of the composite system $\oo+\rr$  in the time period $t \in (t_{k-1}, t_{k})$.  If the system is thermally isolated, then this will be accompanied by the time-ordered sequence of unitary operators $\{U_{\Delta t_k}=e^{-i \Delta t_k H_k}|k \in \{1,\dots, N\}\}$
where the time duration is defined as $\Delta t_k:= t_{k}-t_{k-1}$. The time-ordered application of these results in the unitary operator $U_\tau$, where $\tau = t_N-t_0$, which determines the total evolution of the system.  
 If we identify $H_0:=H_\oo  +  H_\rr$ as the Hamiltonian of the system at times  prior to $t_0$ and posterior to $t_N$,  whereby  the new sequence of Hamiltonians can be aptly called a \emph{Hamiltonian cycle}, then 
\begin{equation}
\Delta Q = \tr[H_\rr (\tr_\oo[U_\tau \rho U_\tau^\dagger]- \rho_\rr(\beta))],
\end{equation}   
will refer to the amount by which the average energy of the reservoir, at times $t>t_N$, will be greater than that at times $t<t_0$, and will have the same meaning as the heat term in \eq{energy change of battery}.   Implicit in this framework is the notion that changing the Hamiltonian acting on the system will take energy from, or put energy into, a work storage device which we do not account for explicitly.   If  a non-unitary evolution is effected, however, we cannot in general make such an identification. This is because a general completely positive, trace preserving map can always be conceived, via Stinespring's dilation theorem \citep{Stinespring}, as resulting from a unitary evolution on the system coupled with an environment. Indeed, the energy consumption in such a case will be determined by the total Hamiltonian of the system plus the environment. If energy is allowed to flow between the system and environment, then the energy increase of $\rr$ (plus the energy increase in $\oo$) will not be identical to the energy consumed from the work storage device; $\Delta Q$ may be less or greater than the energy lost. 

 The only exception to this rule is when the unitary evolution between system and environment conserves the energy of the two individually, whereby no energy is transferred amongst them. This will result in the system to undergo pure dephasing with respect to the (time-local) Hamiltonian eigenbasis; we refer to such a generalised evolution as an energy conserving one. The simplest realisation of such a scenario would require us to  consider the sequence of Hamiltonians to be accompanied by the time-ordered sequence of  super-operators $\{e^{\Delta t_k \louv_k}| k \in \{1,\dots, N\}\}$, with the Liouville super-operators $\louv_k$ defined as
\begin{equation}
\louv_k: \rho \mapsto i[\rho,H_k]_-+\Gamma\sum_{n=1}^{d_\oo d} \left(\pr{\phi_n^k} \rho \pr{\phi_n^k} -\frac{1}{2}[\rho, \pr{\phi_n^k}]_+ \right),\label{reset Lindblad}
\end{equation}
where $\{\ket{\phi_n^k}\}_n$ is the eigenbasis of $H_k$, while  $[\cdot,\cdot]_-$ and $[\cdot,\cdot]_+$ are the commutator and anti-commutator respectively, and $\Gamma \in [0,\infty)$ is the dephasing rate. In each time period $t \in (t_{k-1},t_k)$ the system evolves as $\rho \mapsto (e^{\Delta t_k\louv_k})(\rho)$ while conserving $H_k$; the system evolves by energy conserving, Markovian dephasing channels. As such channels are unital, they will cause the consequent heat dissipation to increase in proportion to the entropy reduction in the object; energy conserving, Markovian dephasing will be detrimental to the erasure process \citep{Reeb-Wolf-Landauer}.

For our two models,  we will  consider the simplest Hamiltonian cycle where the sequence of Hamiltonians sandwiched by $H_0$ is the singleton $\{H_1\}$. Furthermore, we set  $U_\tau=e^{-i \tau H_1}$ to equal $U_\mathrm{opt}^\mathrm{p}(0)$, as determined by the sequential swap algorithm given in \sect{trade-off relation},   when $\tau=1$.  
 Now, we let the system evolve instead  as  $\rho\mapsto(e^{\tau\louv_1})(\rho)$. By again evolving the system for a period of $\tau=1$, we may ascertain how such an environmental interaction affects both the probability of qubit erasure, and the  heat dissipation.

\fig{reset-spin-chain-dephasing} shows the effect of dephasing on the erasure process, when the reservoir is a spin chain of length $N \in \{2,3,4,5\}$, with $\Theta=J=\beta=1$. Similarly, \fig{reset-landauer-dephasing} shows the effect of dephasing on the erasure process, when the reservoir is a ladder system with $\omega=1$ and $\beta=5$. In both instances, an increase in $\Gamma$ results in a  decrease in $p(\varphi_1|\rho_\oo')$ and, with the exception of $N=2$ and $d \leqslant 4$, an increase in $\Delta Q$.   However, not all ladder dimensions $d$, or spin chain lengths $N$, are affected the same way.

We note that when the two reservoirs are dimensionally equivalent, i.e., when  the ladder system has dimensions $d \in \{2^2,2^3,2^4,2^5\}$, commensurate with spin chains of length $N \in \{2,3,4,5\}$,  they display the same behaviour under energy conserving, Markovian dephasing channels. This is because the generator of their evolution, the Liouville super-operator $\louv_1$, is the same in such cases. In both instances, an increase in dimension leads to an increase in $\Delta Q$, while the   probability of qubit erasure increases as we move from $d= 2^2$ to $d=2^3$, decreasing again as we increase further  still to $d=2^4$ and $d=2^5$.  

What is most striking, however, is that the ladder system seems to perform the best precisely when it is dimensionally equivalent to a spin chain.  Consider \fig{dephasing-heat-dimension} and \fig{dephasing-probability-dimension}. Here, $\Delta Q$ and $p(\varphi_1|\rho_\oo')$  are calculated for dimensions $d \in \{2,\dots,32\}$, while keeping all other parameters constant. For dimensions $d \in \{2^2,2^3,2^4,2^5\}$,  we observe that $\Delta Q$ is smaller than that for all larger $d$, while $p(\varphi_1|\rho_\oo')$ attain the largest global values.  As such, we make the following conjecture: 

\
\begin{con}
Let the reservoir be given by a  $d$-dimensional ladder system with a constant energy spacing $\omega$. In the presence of energy conserving, Markovian dephasing, reservoirs with  $d=2^n$, with $n \in \mathds{N}$ and $n\geqslant 2$, allow for the largest global probabilities of qubit erasure while, at the same time, dissipating less heat than all such reservoirs of larger dimension.
\end{con}

\subsection{Full erasure of a  qudit with a  harmonic oscillator}\label{qudit erasure calculations}

Here, we expound on the example of using a ladder system as a reservoir, but consider what happens as we take the limit of infinitely large $d$. In this limit we may call the ladder system a harmonic oscillator. Let us first consider the case where the object is a qudit, with Hilbert space $\h_\oo \simeq \co^{d_\oo}$, prepared in the maximally mixed state
\begin{equation}
\rho_\oo= \frac{1}{d_\oo}\sum_{l=1}^{d_\oo} \pr{\varphi_l}.
\end{equation}

In \app{qudit erasure calculations appendix} we show that the heat dissipation when the reservoir is a  harmonic oscillator is
\begin{align}
 \lim_{d \to \infty}\Delta Q &=  \frac{\omega(d_\oo-1)}{2}\coth\left(\frac{\beta \omega}{2}\right)> \frac{(d_\oo-1)}{\beta}.
\end{align}
 $\Delta Q$ approaches $(d_\oo-1)k_B T$ in the limit as $\omega$ becomes vanishingly small, whereby the spectrum of $H_\rr$ will be approximately continuous.

\begin{figure}[!htb]
\centering
\includegraphics[width=8 cm]{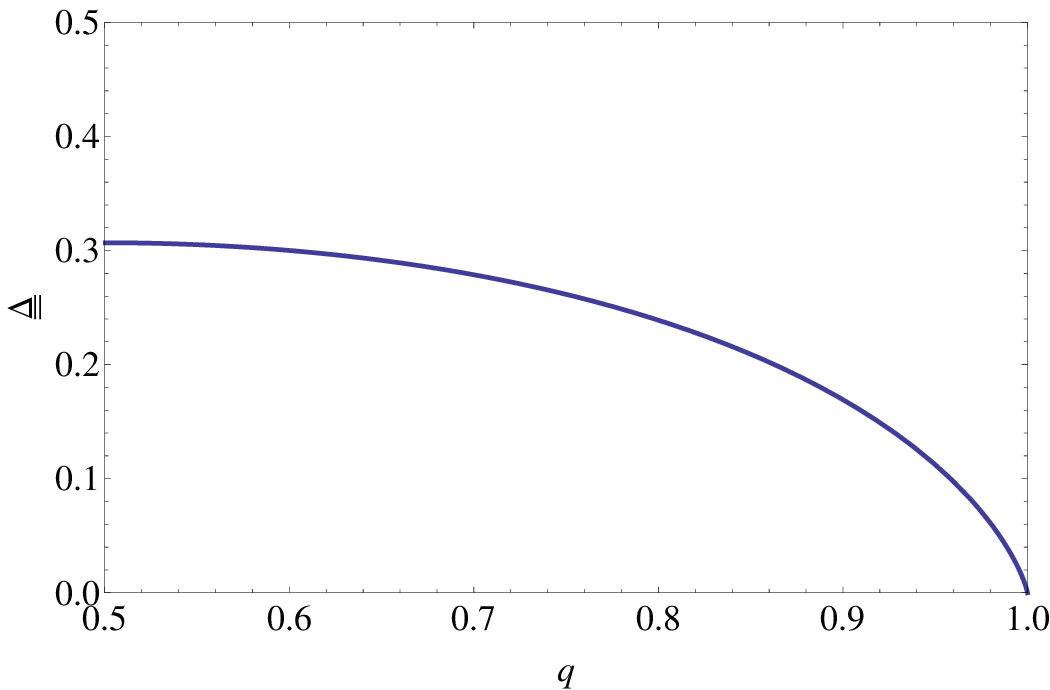}
\caption{ The difference between $\beta\Delta Q$ and $\Delta S$, denoted $\underline{\underline \Delta }$, as a function of the initial bias in the qubit state, $q$. The two coincide only in the trivial case of $q =1$, commensurate with $\Delta S= \Delta Q=0$.   }
\label{Log-base}
\end{figure}

Now let us focus on the case where the object is a qubit, but with an initial bias in its spectrum:  \begin{equation}
\rho_\oo = q \pr{\varphi_1}+(1-q) \pr{\varphi_2} \ , \ q \in \left[\frac{1}{2},1\right).
\end{equation} 
 In \app{general qubit erasure appendix} we show that, in the limit as $\omega$ tends to zero,  $\Delta Q$ will  be
\begin{align}
\Delta Q &= \frac{2q(1-q)\log(\frac{q}{1-q})}{\beta(2q -1)}.
\end{align}
 In the limit as $q$ tends to one-half, $\Delta Q$ approaches $k_B T$ as in our previous analysis.  The concomitant entropy reduction is, of course, always 
\begin{equation}
\Delta S=q \log\left(\frac{1}{q}\right)+ (1-q) \log\left( \frac{1}{1-q}\right).
\end{equation}
By defining the function
\begin{equation}
\underline {\underline\Delta}:= \beta \Delta Q - \Delta S, \end{equation}
as shown in \fig{Log-base}, it is evident that except for the trivial case of $q=1$, commensurate with $\Delta S=\Delta Q=0$, the heat dissipation will exceed Landauer's limit.

\section{Information erasure beyond Landauer's framework}\label{resetting with correlations}

In \sect{optimal unitary section} the  setup for information erasure  had the compound system  of object and thermal reservoir -- our system of interest -- as a thermally isolated quantum system whose constituent parts are initially uncorrelated. The system then undergoes a cyclic process described by a unitary operator, and the average energy increase of the reservoir is defined as heat. Indeed, these are the basic assumptions under which Landauer's principle holds.   To achieve heat dissipation lower than that discussed in \sect{optimal unitary section} we must operate outside of Landauer's framework by abandoning some of these assumptions.  However,  dissipating less heat than Landauer's limit will become meaningless if there is no referent of heat or temperature in the mathematical model. As such, if we wish to avoid making category errors, there are restrictions on the ways in which we may change our assumptions. That is to say, the model must continue to involve a  system that is initially prepared in a Gibbs state that is uncorrelated from any other system considered.  This way, the system has a well-defined temperature, and we may continue to consider its energy increase  as heat. In addition, the process must still be cyclic, i.e., the Hamiltonian of the total system -- in particular the thermal system -- must be the same at the end of the process, as it was at the beginning. If this condition is not satisfied,  we may observe any value of heat we desire by appropriately changing the final Hamiltonian.   

One option available is to move beyond unitary evolution.  This can be achieved by introducing an  auxiliary system to the setup introduced in \sect{optimal unitary section} so that the unitary evolution of the totality results in the object and reservoir to evolve non-unitarily; the auxiliary system must also have a trivial Hamiltonian, proportional to the identity, for the resultant decrease in $\Delta Q$ to always translate to a decrease in energy consumption.   Although the reservoir must always be uncorrelated from the other subsystems for it to be thermalised relative to them \citep{relative-thermalisation}, the auxiliary system and object may  have initial correlations. Unless these correlations are classical, then the resulting dynamics of the object plus reservoir subsystem would cease to be described by completely positive maps  \cite{CP-map-discord,vanishing-discord-CP-map}.  

The other option available is to first consider  a system that  is in a thermal state and, therefore, has a temperature.  Subsequently, the system may be (conceptually) partitioned into two correlated subsystems, with one of them taking the role of the object.  The energy generation due to information erasure of the object, of course, must then be determined over the total system itself. This is because the subsystems do not have well defined Hamiltonians. Although  there is technically no  thermal reservoir to speak of, since the total system was initially thermal, the average energy change thereof may still be called heat in a consistent manner as before.

\subsection{Information erasure with the aid of an auxiliary system}\label{information erasure with an auxiliary system}

\begin{figure}[!htb]
\centering
\includegraphics[width=15 cm]{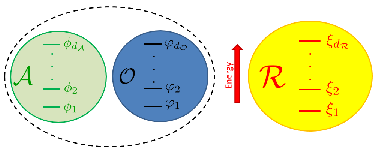}
\caption{ The augmentation of the basic setup by the inclusion of a third, auxiliary system $\aa$  with Hilbert space $\h_\aa\simeq\co^{d_\aa}$. As before, the reservoir is initially in a thermal state and uncorrelated from the rest of the system. The initial state of the object and auxiliary, however, may or may not be correlated.  }
\label{auxiliary}
\end{figure}

Consider a  system composed of: the object, $\oo$, with Hilbert space $\h_\oo \simeq \co^{d_\oo}$; the auxiliary system, $\aa$, with Hilbert space $\h_\aa \simeq \co^{d_\aa}$; and the thermal reservoir, $\rr$, with Hilbert space $\h_\rr \simeq \co^{d_\rr}$.   Let the initial state of the system be $\rho \otimes \rho_\rr(\beta)$, with  $\rho$ the state of $\oo+\aa$ and $\rho_\rr(\beta)$  the Gibbs state of the thermal reservoir. This setup is represented diagrammatically in \fig{auxiliary}. We may (probabilistically) prepare the object in a pure state by conducting a cyclic process on the total system, characterised by a unitary operator, as before.  By letting the Hamiltonian of the auxiliary system, $H_\aa$, be proportional to the identity, we may ensure that  the total energy consumption due to this process would be accounted for by the energy change of the object and thermal reservoir alone. As before, the energy change of the thermal reservoir, $\Delta Q$, is heat. 
 
In the extreme case, we may consider that the unitary operator  acts non-trivially only on the object plus auxiliary subsystem; the thermal reservoir will thus not be involved, and no heat will be dissipated. We would like to know what the necessary and sufficient conditions for fully erasing the object would be in this case. The mapping $\rho_\oo \mapsto \rho_\oo':= \tr_\aa[U \rho U^\dagger]$, with $U $ acting on $\h_\oo\otimes \h_\aa$, will fully erase $\oo$ into the pure state $\ket{\varphi
_1}$ if and only if  
\begin{equation}
U\rho U^\dagger= \pr{\varphi_1}\otimes \sum_{n=1}^{ R_\aa \leqslant d_\aa} q_n^\downarrow\pr{\phi_n},
\end{equation}
where $R_\aa$ is the rank of $\rho_\aa'$ and, hence, the rank of $\rho$.  The class of states that allow for such a transformation can, without loss of generality, be represented as \begin{equation}
\rho = \sum_{n=1}^{R_\aa\leqslant d_\aa} q_n^\downarrow U^\dagger(\pr{\varphi_1} \otimes \pr{\phi_n})U.
\end{equation}
Therefore, a necessary and sufficient condition for full information erasure by  unitary evolution, without using the thermal reservoir, is for the rank of $\rho$ to be less than, or equal to, $d_\aa$.  
To see how correlations between $\oo$ and $\aa$ come into play, consider the simple case where $d_\oo=d_\aa=2$ and, for $\lambda \in (1/2,1)$, the following states:
\begin{align}
\rho_{\mathrm{u.c.}}&=(\lambda\pr{\varphi_1}+(1-\lambda)\pr{\varphi_2})\otimes \pr{\phi_1},\nonumber \\
\rho_{\mathrm{c.c.}}&=\lambda \pr{\varphi_1}\otimes \pr{\phi_1}+(1-\lambda)\pr{\varphi_2}\otimes \pr{\phi_2}, \nonumber \\
\rho_{\mathrm{q.d.}}&=\lambda \pr{\varphi_1}\otimes \pr{\phi_1}+(1-\lambda)\pr{\varphi_2}\otimes \pr{\phi_+}, \ \ \ \ket{\phi_\pm}:=\frac{1}{\sqrt{2}}(\ket{\phi_1}\pm \ket{\phi_2}),\nonumber \\
\rho_{\mathrm{p.e.}}&=\pr{\psi}, \ \ \ \ket{\psi}=\sqrt{\lambda}\ket{\varphi_1}\otimes \ket{\phi_1}+\sqrt{1-\lambda}\ket{\varphi_2}\otimes \ket{\phi_2}. 
\end{align}
All of these have a rank of at most 2, and   the  reduced state $\rho_\oo= \lambda\pr{\varphi_1}+(1-\lambda)\pr{\varphi_2}$ which, with an appropriate unitary operator, can be fully erased to $\pr{\varphi_1}$. Each state, however, falls under a different class of correlations: $\rho_{\mathrm{u.c.}}$ is uncorrelated, $\rho_{\mathrm{c.c.}}$ is classically correlated,  $\rho_{\mathrm{q.d.}}$ has quantum discord, and   $\rho_{\mathrm{p.e.}}$ is a  pure entangled state. The only case where the state of $\aa$ is also left intact, however, is when the two systems are classically correlated. Notwithstanding, this cannot be seen as allowing for $\aa$ to act as a catalyst for information erasure. For $\aa$ to be utilised in the information erasure of another object system, with the same unitary operator, the two must first be correlated; this process  will have a thermodynamic cost itself \citep{correlation-thermo-cost}. In the case where $\oo$ and $\aa$ are in a pure entangled state, the unitary operator which prepares $\oo$ in a pure state will also prepare $\aa$  in a pure state. As discussed in \citep{negative-entropy}, this will allow for $\rr$ to be cooled by transferring entropy from it to $\aa$, resulting in a negative $\Delta Q$.  

In either scenario, the initial   state $\rho$ on the composite system of $\oo+\aa$, which has a rank smaller than $d_\aa$, can be seen as a thermodynamic resource. This is  because it is a system that is highly out of equilibrium. Recall that the Hamiltonian of $\aa$ is considered to be trivial, being proportional to the identity. As such, if this system was also at thermal equilibrium with the inverse temperature $\beta$, then we would have $\rho=\rho_\aa(\beta)\otimes \rho_\oo=\frac{1}{d_\aa} \one_\aa \otimes \rho_\oo$. Any unitary operator acting on such a system would not be able to increase the largest eigenvalue of $\rho_\oo$. As such, information erasure would not be possible.

In the case where the rank of $\rho$ is greater than $d_\aa$, but smaller than $d_\oo d_\aa$, the reservoir may be used to facilitate information erasure using similar arguments as in \sect{optimal unitary section}. This will allow for a  larger $p_{\varphi_1}^{\max}$, and a smaller consequent $\Delta Q$,   than if $\aa$ was not present.

\subsection{Object  as a component of a thermal system}

\begin{figure}[!htb]
\centering
\includegraphics[width=10 cm]{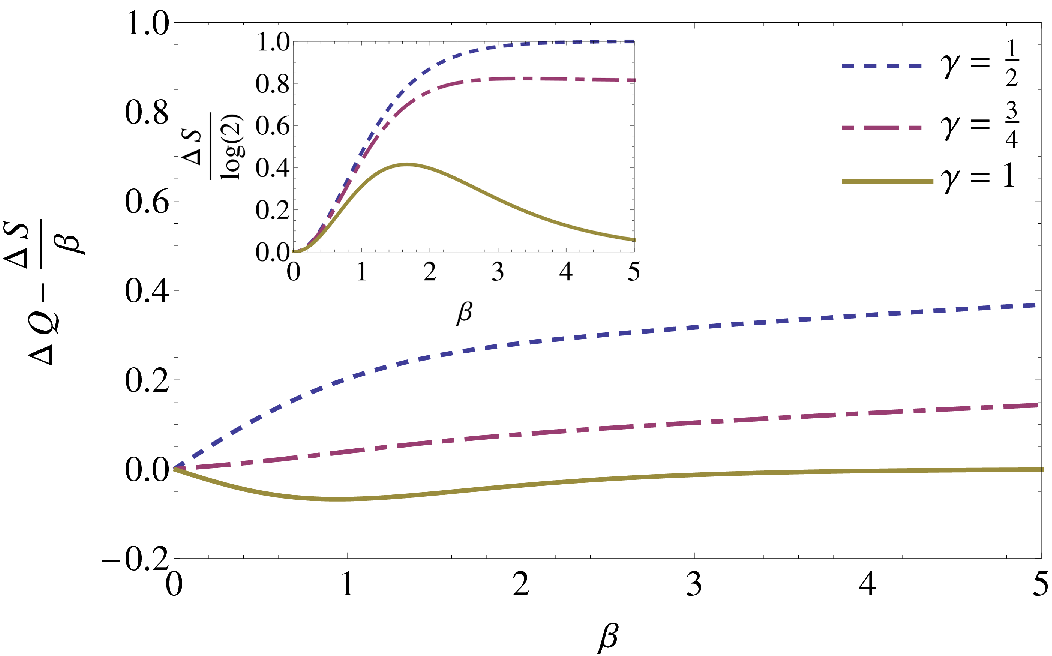}
\caption{(Colour online) Optimal information erasure of system $\oo$, with the composite system of $\oo+\kk$  initially in a thermal state. As the entanglement in the eigenvectors of the Hamiltonian vanishes, where $\gamma \to 1$, both $\Delta Q$ and $\Delta S$ decrease. This is done, however, in such a way that   $ \Delta Q - \frac{1}{\beta}\Delta S$ becomes negative at intermediate temperatures, thereby resulting in the ``violation'' of Landauer's limit.  }
\label{correlated-reservoir}
\end{figure}      

Consider a system composed of an object, $\oo$, with Hilbert space $\h_\oo\simeq \co^{d_\oo}$, and some other system, $\kk$, with Hilbert space $\h_\kk\simeq \co^{d_\kk}$.  The composite system has the Hamiltonian 
\begin{equation}
H=\sum_{n=1}^{d_\oo d_\kk} \lambda_n^\uparrow \pr{\xi_n}. 
\end{equation}
Let the initial state of the system be in the thermal state $\rho(\beta)=e^{-\beta H}/\tr[e^{-\beta H}]\equiv\sum_n p_n^\downarrow \pr{\xi_n}$ with the non-increasing vector of probabilities $\bm{p^\downarrow}:=\{p_n^\downarrow\}_n$. We wish to prepare the subsytem $\oo$ in the pure state $\ket{\Psi}$ by some cyclic process characterised by a unitary operator $U$ acting on $\h_\oo \otimes \h_\kk$. By \lemref{maximum-probability} the maximal probability of achieving this is accomplished by choosing $U$ from an equivalence class of unitary operators $[U_\mathrm{maj}]$ characterised by the rule 
\begin{equation}
U_\mathrm{maj}:\begin{cases}\ket{\xi_n} \mapsto \ket{\Psi}\otimes \ket{\phi_j} & \text{ if } n \in \{1,\dots,d_\kk\}, \\
\ket{\xi_n} \mapsto \ket{\nu_k} & \text{ if } n \notin \{1,\dots,d_\kk\}. \\
\end{cases}\label{correlated systems maximal information erasure unitary}
\end{equation}
  Each of the vectors in $\{\ket{\nu_k}\in \h_\oo \otimes \h_\kk\}_k$ are orthogonal to those in $\{\ket{\Psi}\otimes \ket{\phi_j}\}_j$, so that the union thereof forms an orthonormal basis that spans $\h_\oo\otimes \h_\kk$. As the system was initially thermal,   the gain in its average energy is heat, which obeys the identity
\begin{align}
\Delta Q:= \tr[H(\rho'-\rho(\beta))]&=\frac{ S(\rho'\|\rho(\beta))+S(\rho')-S(\rho(\beta))}{\beta}= \frac{1}{\beta}S(\rho'\| \rho(\beta)). \end{align}
Here, we make the substitution $\rho':=U \rho(\beta) U^\dagger$.  As unitary evolution does not alter the von Neumann entropy, this energy production is a function of the relative entropy alone;
$\Delta Q$ is therefore nonnegative and  independent of $\Delta S:= S(\rho_\oo)- S(\rho_\oo')$. To determine how $\Delta Q$ can be minimised, we write $S(\rho'\| \rho(\beta))$ in the alternative way
\begin{align}
S(\rho'\| \rho(\beta)) &=  \sum_{n=1}^{d_\oo d_\kk} q^U_n \log\left(\frac{1}{p_n^\downarrow}\right)-S(\rho'),\nonumber \\
&= \bm{q^U} \cdot \bm{\mathrm{log}_p^\uparrow} - S(\rho(\beta)),
\label{relative entropy term}\end{align} 
where $\bm{q^U}$ is a vector of real numbers
\begin{equation}
q^U_n:=\sum_{m=1}^{d_\oo d_\kk}p_m^\downarrow|\<\xi_n|U|\xi_m\>|^2,\label{relative entropy term probabilities}
\end{equation}
and $\bm{\mathrm{log}_p^\uparrow}:= \{\log(1/p_n^\downarrow)\}_n$ a non-decreasing vector. In \app{Object as a component of a thermal system proof} we prove that to minimise $\Delta Q$, after having maximised the probability of preparing $\oo$ in the pure state $\ket{\Psi}$, we must choose the unitary operator $U \in [U_\mathrm{maj}^1] \subset [U_\mathrm{maj}]$  so that $\bm{q^U}$ is non-increasing  with respect to the energy eigenvalues of $H$, and that it majorises all such possible vectors. If we are free to choose what Hamiltonian to construct for the system, then the heat dissipation may be further minimised by choosing the eigenvectors of $H$, that have support on $\{\ket{\Psi} \otimes \ket{\phi_j}\}_j$, to be chosen from this set itself. In other words, the optimal value of $\Delta Q$ is achieved when the eigenvectors of $H$ are uncorrelated with respect to the $\oo:\kk$ partition.

Let us  consider a simple example, where $d_\oo=d_\kk=2$, and the eigenvectors of $H$ are given as
\begin{align}
\ket{\xi_1}&=\sqrt{\gamma_+}\ket{\varphi_1}\otimes\ket{\phi_1} +\sqrt{1-\gamma_+}\ket{\varphi_2}\otimes\ket{\phi_2}, \nonumber \\
\ket{\xi_2}&=\sqrt{1-\gamma_+}\ket{\varphi_1}\otimes\ket{\phi_1} -\sqrt{\gamma_+}\ket{\varphi_2}\otimes\ket{\phi_2},  \nonumber\\
\ket{\xi_3}&=\sqrt{\gamma_-}\ket{\varphi_1}\otimes\ket{\phi_2} +\sqrt{1-\gamma_-}\ket{\varphi_2}\otimes\ket{\phi_1}, \nonumber\\
\ket{\xi_4}&=\sqrt{1-\gamma_-}\ket{\varphi_1}\otimes\ket{\phi_2} -\sqrt{\gamma_-}\ket{\varphi_2}\otimes\ket{\phi_1}. \end{align} 
Moreover, let   $\lambda_1^\uparrow=0$ and $\lambda_{n+1}^\uparrow-\lambda_n^\uparrow=1$ for all $n$. Conforming with \eq{correlated systems maximal information erasure unitary}, the unitary operator
\begin{align}
U:\begin{cases} 
\ket{\xi_1} \mapsto \ket{\Psi} \otimes \ket{\phi'_1},  \\
\ket{\xi_2} \mapsto \ket{\Psi} \otimes \ket{\phi'_2},  \\
\ket{\xi_3} \mapsto \ket{\Psi^\perp} \otimes \ket{\phi''_1},  \\
\ket{\xi_4} \mapsto \ket{\Psi^\perp} \otimes \ket{\phi''_2},  \\
\end{cases}\label{correlated unitary}
\end{align}
will then prepare $\oo$ in the state $\ket{\Psi}$,  with  $p^{\max}_\Psi=p_1^\downarrow+p_2^\downarrow$.   We may then  minimise $\Delta Q$ by choosing $U$  from the equivalence class $[U_\mathrm{maj}^1]$.  This can be achieved if we choose  the vectors $\ket{\Psi}$, $\ket{\phi_1'}$, and $\ket{\phi_1''}$ respectively from the sets $\{\ket{\varphi_1},\ket{\varphi_2}\}$, $\{\ket{\phi_1},\ket{\phi_2}\}$, and  $\{\ket{\phi_1},\ket{\phi_2}\}$; what particular permutation does this depends on the temperature and the values of $\gamma_\pm$. In the special case of $\gamma_+=\gamma_-=\gamma $, for example, we find that irrespective of the temperature, when $\gamma> 1/2$ this is achieved when $\ket{\Psi}=\ket{\varphi_1}$, $\ket{\phi_1'}=\ket{\phi_1}$, and $\ket{\phi_1''}=\ket{\phi_2}$. Conversely when $\gamma<1/2$ this is realised when  $\ket{\Psi}=\ket{\varphi_2}$, $\ket{\phi_1'}=\ket{\phi_2}$, and $\ket{\phi_1''}=\ket{\phi_1}$. 

\fig{correlated-reservoir} shows the dependence of both $\Delta S$ and $ \Delta Q-\Delta S/\beta$ on the entanglement of the Hamiltonian eigenvectors $\{\ket{\xi_n}\}_n$, with $\gamma_+=\gamma_-=\gamma \in \{1/2,3/4,1\}$. The system is always evolved  using  $U \in [U_\mathrm{maj}^1]$. As $\gamma$ tends to one,  thereby resulting in  uncorrelated Hamiltonian eigenvectors,   both  $\Delta Q$  and  $\Delta S$ decrease,  vanishing in the  the limit as $\beta$ tends to infinity. However, for intermediate temperatures, $\Delta Q$ becomes so low that it ``violates'' Landauer's limit. This is similar to the possibility of extracting work from the correlations between a quantum system and its environment, which are initially in a thermal state \citep{Gelbwaser-work-extraction-QND}.

\section{Conclusions}\label{conclusions}

In this article, we  have developed a context-dependent, dynamical variant of Landauer's principle. We used techniques from majorisation theory to characterise the equivalence class of unitary operators that bring the probability of information erasure to a desired value and minimise the consequent heat dissipation to the thermal reservoir. By constructing a sequential swap algorithm, we demonstrated that there is a tradeoff between the probability of information erasure and the minimal heat dissipation. Furthermore, we showed that except for the cases where the object is a two-level system, or when we are able to fully erase the object's information, we may maximise the probability of information erasure without also minimising the object's entropy; this allows for a more energy-efficient procedure for probabilistic information erasure.

We also investigated methods of reducing heat dissipation due to information erasure by operating outside of Landauer's framework. However, we wanted this departure to preserve the referent of heat and temperature in our mathematical description; dissipating less heat than Landauer's limit becomes meaningless when there is no temperature or  heat to speak of.  Therefore, we arrived at two alterations to Landauer's framework which would not result in a category error with respect to heat and temperature. The first avenue was to introduce an auxiliary system to the object and reservoir, while the second was to consider the object as a subpart of a system in thermal equilibrium. In the first instance, the figure of merit was identified as the rank of the system in the object-plus-auxiliary subspace; if the rank of this state is less than the dimension of the auxiliary Hilbert space, then full information erasure is possible with at most zero heat dissipation to the reservoir. In the second instance, information erasure can be achieved with possibly less heat than Landauer's limit when the eigenvectors of the system Hamiltonian, that have support on the pure state we which to prepare the object in, are product states. 

The primary question we have not addressed in this study, and shall leave for future work,  is the inclusion of time-dynamics into what we consider as the physical context; the optimal unitary operator for information erasure is considered here as a bijection between orthonormal basis sets.  In most realistic settings, however, one is restricted in the Hamiltonians they can establish between the object and reservoir. As such, the optimal unitary operator may not always be reachable, resulting in a smaller maximal probability of information erasure, a larger minimal heat dissipation, or both. Furthermore, an interesting question to address is the number of times  that we must switch between the Hamiltonians, that generate the unitary group, in order to  obtain the optimal unitary operator, and how  this would scale with the reservoir's dimension. This would provide a link between the present work and the third law of thermodynamics \citep{Masanes-third-law} from a control-theoretic \citep{Domenico-control-dynamics} viewpoint.   

\section*{Acknowledgments}
 The authors wish to thank Ali Rezakhani for his useful comments on the manuscript. M. H. M. and Y. O. thank the support from Funda\c{c}\~{a}o para a Ci\^{e}ncia e a Tecnologia (Portugal), namely through programmes PTDC/POPH and projects UID/Multi/00491/2013, UID/EEA/50008/2013, IT/QuSim and CRUP-CPU/CQVibes, partially funded by EU FEDER, and from the EU FP7 projects LANDAUER (GA 318287) and PAPETS (GA 323901).   Furthermore M. H. M. acknowledges the support from the EU FP7 Marie Curie Fellowship (GA 628912).

\makeatletter
\renewcommand\p@subsection{\thesection\,}
\makeatother
\makeatletter
\renewcommand\p@subsubsection{\thesection\,\thesubsection\,}
\makeatother

\appendix

\section{Majorisation theory}\label{Majorisation theory section}      
    
Here we shall introduce some useful concepts from the theory of majorisation \citep{Majorisation}.  Given a vector of real numbers $\bm{a}:=\{a_i\}_{i} \in \re^N$, where $i$ runs over the index set $I:=\{1,\dots,N\}$,   we may construct the ordered vectors $\bm{a^{\uparrow}}:= \{a_i^{\uparrow}\}_i$ and $\bm{a^{\downarrow}}:= \{a_i^{\downarrow}\}_i$  by permuting the elements in $\bm{a}$.  The non-decreasing vector $\bm{a}^{\uparrow}$  is defined such that for all $i,j \in I$ where $i<j$, we have   $a_i^{\uparrow} \leqslant a_{j}^{\uparrow}$.  Conversely the non-increasing vector $\bm{a}^{\downarrow}$ is defined such that for all $i,j \in I$ where $i<j$, we have $a_i^{\downarrow} \geqslant a_{j}^{\downarrow}$. The vector $\bm{a}$ is said to be weakly majorised by $\bm{b}$ from below, denoted $\bm{a} \prec_w \bm{b}$, if and only if for every $k \in I$, $\sum_{i=1}^k b_i^{\downarrow}\geqslant \sum_{i=1}^k a_i^{\downarrow}$. Conversely,  $\bm{a}$ is said to be weakly majorised by $\bm{b}$ from above, denoted $\bm{a} \prec^w \bm{b}$, if and only if for every $k \in I$, $\sum_{i=1}^k a_i^{\uparrow}\geqslant \sum_{i=1}^k b_i^{\uparrow}$. The stronger condition of $\bm{a}$ being majorised by $\bm{b}$, denoted $\bm{a} \prec \bm{b}$, is satisfied if both $\bm{a} \prec_w \bm{b}$ and $\bm{a} \prec^w \bm{b}$ (or alternatively, if one of these conditions is met together with   $\sum_i a_i=\sum_i b_i$).  A  sufficient condition for $\bm{a} \prec_w \bm{b}$ is if for all $i\in I$, $a_i^{\downarrow} \leqslant b_i^{\downarrow}$. Similarly, a sufficient condition for    $\bm{a} \prec^w \bm{b}$ is if for all $i\in I$, $a_i^{\uparrow} \geqslant b_i^{\uparrow}$. We now introduce a theorem that will be central to many results in this article.

\
\begin{thm}\label{monotone-lattice-superadditive}
For two vectors $\bm{a}$ and $\bm{b}$, in $\re^N$, their inner product obeys the relation
\begin{equation*}
\bm{a}^{\downarrow} \cdot \bm{b}^{\uparrow} \leqslant \bm{a} \cdot \bm{b} \leqslant \bm{a}^{\downarrow} \cdot \bm{b}^{\downarrow}. 
\end{equation*}
\end{thm}
For a proof we refer to Theorem II.4.2 in \citep{Matrix-Analysis}. This leads to the simple corollary:

\
\begin{cor}\label{corollary-monotone-lattice-superadditive}
Consider the pairs of vectors  $\{\bm{a_1},\bm{a_2}\}$ and $\{\bm{b_1},\bm{b_2}\}$,  such that $\bm{a_1} \prec \bm{a_2}$ and $\bm{b_1} \prec \bm{b_2}$.  It follows from \thmref{monotone-lattice-superadditive} that $\bm{a_1}^\downarrow \cdot \bm{b_1}^\downarrow  \leqslant \bm{a_2}^\downarrow \cdot \bm{b_2}^\downarrow$, and $\bm{a_2}^\downarrow \cdot \bm{b_2}^\uparrow  \leqslant  \bm{a_1}^\downarrow \cdot \bm{b_1}^\uparrow$. 
\end{cor}

\section{An Equivalence class of unitary operators}\label{equivalence class of unitary operators}      
Here we expound on the sense in which  \eq{maximising probability rule equation}  characterises an equivalence class of unitary operators, instead of just one unique unitary operator. The arguments here translate to the other equivalence classes of unitary operators mentioned in the article. Here, two unitary operators $U$ and $V$ are said to be equivalent so far as the probability of information erasure is concerned, if and only if 
\begin{equation}
\tr[(\pr{\varphi_1}\otimes \one)U\rho U^\dagger]=\tr[(\pr{\varphi_1}\otimes \one)V \rho V^\dagger].
\end{equation}

   First of all, a degeneracy in the probability distribution $\bm{p^\downarrow}$ will mean that the representation of $\rho$, as given in \eq{composite system state representation initial}, will not be unique. If, for example, we have $p_i^\downarrow=p_j^\downarrow$, then $\ket{\psi_i}$ and $\ket{\psi_j}$ in \eq{composite system state representation initial} can be replaced by any orthonormal pair of vectors $\{\ket{\psi},\ket{\psi ^\perp}\}$ that span the same subspace. As such, replacing $U_\mathrm{maj}^g\ket{\psi_i}=  \ket{\varphi_1} \otimes\ket{ \xi_i'}$ with $U_\mathrm{maj}^g\ket{\psi}=  \ket{\varphi_1} \otimes\ket{ \xi_i'}$, and similarly  $U_\mathrm{maj}^g\ket{\psi_j}=  \ket{\varphi_1} \otimes\ket{ \xi_j'}$ with $U_\mathrm{maj}^g\ket{\psi^\perp}=  \ket{\varphi_1} \otimes\ket{ \xi_j'}$, would give a different unitary operator, but the same probability of information erasure. As such, both unitary operators belong in the same equivalence class with respect to the probability of information erasure, denoted $[U_\mathrm{maj}]$. Additionally, as the probability of information erasure is unaffected by the orthonormal basis $\{\ket{\xi_m'}\}$ in the transformation rules of \eq{maximising probability rule equation}, then any choice of this basis will define a different unitary operator that, nonetheless, belongs to the same equivalence class $[U_\mathrm{maj}]$.

\section{Technical proofs}\label{technical proofs}

\subsection{Maximising the probability of information erasure}\label{proof maximising probability of information erasure}
 Recall that the probability of preparing  $\rho_\oo'$ in the state $\ket{\varphi_1}$ is defined as 
\begin{align}
p(\varphi_1| \rho_\oo'):= \<\varphi_1|\rho_\oo'|\varphi_1\> &= \sum_{n=1}^{d_\oo d_\rr}  p_n ^{\downarrow} \<\psi_n|U^\dagger(\pr{\varphi_1}\otimes \one_\rr)U|\psi_n\>, \nonumber \\ & = \sum_{n=1}^{d_\oo d_\rr}  p_n^{\downarrow} g_n(U) \equiv\bm{p}^{\downarrow}\cdot\bm{g(U)},
\end{align} 
where $\bm{g(U)}$ is a vector of positive numbers $g_n(U):= \<\psi_n|U^\dagger(\pr{\varphi_1}\otimes \one_\rr)U|\psi_n\>$ such that $\sum_n g_n(U)=d_\rr$.  

\
\begin{lem}\label{maximum-probability}
The maximum probability of information erasure is $p_{\varphi_1}^{\max}= \sum_{m=1}^{d_\rr} p_m^{\downarrow}$. The equivalence class of unitary operators that achieve this, denoted $[U^g_\mathrm{maj}]$,   is characterised by the rule
\begin{equation}
\text{for all } m\in \{1, \dots, d_\rr\} \ , \ U_\mathrm{maj}^g\ket{\psi_m}=  \ket{\varphi_1} \otimes\ket{ \xi_m'},
\end{equation}
where $\{\ket{\xi_m'}\}_m$ is an arbitrary orthonormal basis in $\h_\rr$.
\end{lem}
\begin{proof}
By \thmref{monotone-lattice-superadditive} we know that $\bm{p}^{\downarrow}\cdot\bm{g(U)} \leqslant \bm{p}^{\downarrow}\cdot\bm{g^{\downarrow}(U)}$. Let $U^g_\mathrm{maj}$ be a member of an equivalence class of unitary operators such that $ \bm{g(U_\mathrm{maj}^g)}= \bm{g^{\downarrow}(U_\mathrm{maj}^g)}$ and $ \bm{g^\downarrow(U)}\prec \bm{g^\downarrow(U_\mathrm{maj}^g)}$ for all $U$ acting on $\h_\oo \otimes \h_\rr$.   Therefore, by \corref{corollary-monotone-lattice-superadditive} we get $\bm{p}^{\downarrow}\cdot\bm{g^{\downarrow}(U)} \leqslant  \bm{p}^{\downarrow}\cdot\bm{g^{\downarrow}(U^g_\mathrm{maj})}$, and hence  $p(\varphi_1|\rho_\oo')$ is maximised by $U^g_\mathrm{maj}$. Because  $g_n(U) \in [0,1]$ for all $n$, and  $\sum_n g_n(U)= d_\rr$, the first $d_\rr$ elements in $\bm{g^{\downarrow}(U^g_\mathrm{maj})}$ must be one, and the rest zero.  
\end{proof}

\subsection{Minimising the heat dissipation}\label{proof Minimising the heat dissipation}

 We may always write the post-transformation marginal state of the reservoir as 
\begin{equation}
\rho_\rr'=\sum_{m=1}^{d_\rr} r_m'^{\downarrow}(U) \pr{\xi_m'},
\end{equation}
with $\bm{r'^{\downarrow}(U)} :=\{r_m'^{\downarrow}(U)\}_m$ a non-increasing vector of probabilities and $\{\ket{\xi_m'}\}_m$ an arbitrary orthonormal basis in $\h_\rr$. Because $\rho_\rr(\beta)$ is fixed,  minimising $\Delta Q$ is achieved by minimising the average energy of this state, given as 
 \begin{align}
\tr[H_\rr \rho_\rr']&=\sum_{m=1}^{d_\rr} r_m'^{\downarrow}(U)\<\xi_m'| H_\rr|\xi_m'\>\equiv \bm{r'^{\downarrow}(U)} \cdot \bm{\lambda'} ,
 \end{align}
where $\bm{\lambda'}$ is a vector of real numbers $\lambda_m':= \<\xi_m'| H_\rr|\xi_m'\> $. To determine how $\Delta Q$ can be minimised, we first provide a recursive proof of the Ky Fan principle \citep{Matrix-Analysis} to show that  the set of eigenvalues $\bm{\lambda}$    majorises    all possible $\bm{\lambda'}$.

\
\begin{lem} \label{recursive-ky-fan} 
$\bm{\lambda'} \prec \bm{\lambda}$ for all orthonormal bases $\{\ket{\xi_m'} \in \h_\rr\}_m$.
\end{lem}
\begin{proof}
To show this, it is sufficient to show that $\sum_m \lambda_m= \sum_m \lambda_m' $ and $\lambda' \prec^w \lambda$ for all $\{\ket{\xi'_m}\}_m$. The first condition is trivial, as $\sum_m \lambda_m' = \tr[H_\rr]$ and is independent of $\{\ket{\xi_m'}\}_m$. To show that $\lambda' \prec^w \lambda$,   it is sufficient to prove that for all $m$ and $\{\ket{\xi_m'}\}_m$, $\lambda_m^{\uparrow}\leqslant \lambda_m'^{\uparrow}$. This can be done by showing that the minimal value attainable by $\lambda_1'^{\uparrow}$ is $\lambda_1^{\uparrow}$ and, given this constraint, the minimal value attainable by $\lambda_2'^{\uparrow}$ is $\lambda_2^{\uparrow}$, and  so on.  One may always write $\ket{\xi_m'}= \alpha_m \ket{\xi_m}+\beta_m \ket{\xi_m ^\perp}$ where $\ket{\xi_m^\perp}$ is the orthogonal complement to $\ket{\xi_m}$ in $\h_\rr$. Consequently, we have $\lambda_m'^{\uparrow}=|\alpha_m|^2 \<\xi_m|H_\rr | \xi_m\> + |\beta_m|^2 \<\xi_m^\perp|H_\rr | \xi_m^\perp\>$. It is evident  that 
$\<\xi_1^\perp|H_\rr |\xi_1^\perp\> \geqslant \<\xi_1|H_\rr |\xi_1\>=: \lambda_1^{\uparrow}$. Therefore we know that $\lambda_1'^{\uparrow}$ is minimised by setting $|\alpha_1|^2=1$.   In the next step, the fact that $\<\xi_1'|\xi_2'\>=0$ and that our previous step sets $\ket{\xi_1'}=\ket{\xi_1}$ implies that $\<\xi_1|\xi_2^\perp\>=0$. This in turn implies that $\<\xi_2^\perp| H_\rr|\xi_2^\perp\>\geqslant \<\xi_2|H_\rr |\xi_2\> =:\lambda_2^{\uparrow}$, so that $\<\xi_2'|H_\rr |\xi'_2\>$ is minimised by setting $|\alpha_2|^2=1$. This argument can be made recursively for all $m$. \end{proof}

Now we are able to characterise the equivalence class of unitary operators that minimise $\Delta Q$.  

\
\begin{lem}\label{minimising-heat}
$\Delta Q$ is minimised by the equivalence class of unitary operators $[U_\mathrm{maj}^f]$ characterised by the rule
\begin{equation*}
\text{for all } m \in \{1, \dots,d_\rr\}  \text{ and } n \in\{(m-1)d_\oo+1,\dots ,m d_\oo\} \ , \ U_\mathrm{maj}^f\ket{\psi_n}= \ket{\varphi_{l}^m} \otimes \ket{\xi_m} ,
\end{equation*}
with the set  $\{\ket{\varphi_{l}^m}|l \in \{1, \dots, d_\oo\}\}$ forming an orthonormal basis in $\h_\oo$ for each $m$. 
\end{lem}
\begin{proof}
By \corref{corollary-monotone-lattice-superadditive} and \lemref{recursive-ky-fan},  $  \bm{r'^{\downarrow}(U)} \cdot \bm{\lambda^{\uparrow}} \leqslant \bm{r'^{\downarrow}(U)} \cdot \bm{\lambda'} $. Therefore $\tr[H_\rr \rho_\rr']$ is minimal when for all $m$, $\ket{\xi_m'}= \ket{\xi_m}$. In such a case, we have 
 \begin{align}
 r_m'^{\downarrow}(U):=\<\xi_m|\rho_\rr'|\xi_m\>&=\sum_{n=1}^{d_\oo d_\rr} p_n^{\downarrow}\<\psi_n|U^\dagger(\one_\oo \otimes \pr{\xi_m})U|\psi_n\>, \nonumber \\
&= \sum_{n=1}^{d_\oo d_\rr}  p_n ^{\downarrow}f_n(U,m) =\bm{p}^{\downarrow}\cdot\bm{f(U,m)},
 \end{align}
  where $\bm{f(U,m)}$ is a vector of positive numbers  $f_n(U,m):= \<\psi_n|U^\dagger(\one_\oo \otimes \pr{\xi_m})U|\psi_n\>$ such that $\sum_n f_n(U,m)=d_\oo$ for all $m$. Let    $U_\mathrm{maj}^f$ be a member of the equivalence class of unitary operators such that $\bm{r'^{\downarrow}(U)}\prec\bm{r'^{\downarrow}(U_\mathrm{maj}^f)} $ for all $U$ acting on $\h_\oo \otimes \h_\rr$. By \corref{corollary-monotone-lattice-superadditive} it would then follow that    $\bm{r'^{\downarrow}(U_\mathrm{maj}^f)} \cdot \bm{\lambda^{\uparrow}}\leqslant\bm{r'^{\downarrow}(U)} \cdot \bm{\lambda^{\uparrow}} $, resulting in the minimisation of $\tr[H \rho_\rr']$ and hence $\Delta Q$. To find $\bm{r'^{\downarrow}(U_\mathrm{maj}^f)}$, we first need to maximise $r_1'^{\downarrow}(U)$ and then, given this constraint, maximise $r_2'^{\downarrow}(U)$, and so on. This, in turn, is achieved by choosing $U_\mathrm{maj}^f$ so that $\bm{f(U_\mathrm{maj}^f,1)}=\bm{f^{\downarrow}(U_\mathrm{maj}^f,1)}$ and $\bm{f^{\downarrow}(U_\mathrm{maj}^f,1)}\succ\bm{f^{\downarrow}(U,1)}$ for all $U$. Note that for each $m$,   $f_n(U,m) \in [0,1]$ for all $n$, and $\sum_nf_n(U,m)=d_\oo$.   Hence, the first $d_\oo$ entries of $\bm{f^{\downarrow}(U_\mathrm{maj}^f,1)}$ are  taken to one and the rest to zero.   Because of the constraint posed by the orthogonality of the vectors $\{U\ket{\psi_n}\}_n$, however, the first $d_\oo$ elements of $\bm{f(U_\mathrm{maj}^f,2)}$ must be zero, and to maximise $r_2'^{\downarrow}(U)$ the best we can do is to only take the second $d_\oo$ entries of $\bm{f(U_\mathrm{maj}^f,2)}$  to one, with the rest being zero. This argument is then made recursively for all $m$.
\end{proof}

\subsection{Minimal heat dissipation conditional on maximising the probability of information erasure }\label{proof maximal-probability-optimal-unitary}

  Let us  divide the vector of probabilities $\bm{p ^{\downarrow}}$ to form the non-increasing vector of cardinality $d_\rr$, denoted $\Pi_0^{\downarrow}$, and the non-increasing vectors of cardinality $d_\oo-1$, denoted $\{\Pi_m^{\downarrow}|m \in \{1,\dots, d_\rr\}\}$,  defined as
\begin{align}
\Pi_0^{\downarrow}&:= \{ p_m^{\downarrow}| m \in\{1,\dots,d_\rr\}\}, \nonumber \\
 \Pi^{\downarrow}_{m\geqslant 1}&:= \{ p^{\downarrow}_{d_\rr + (m-1)( d_\oo-1) +l} | l \in \{1, \dots, d_\oo-1\}\}.
\end{align}
 We refer to the $m\ts{th}$   element of $\Pi_0^{\downarrow}$ as $\Pi_0^{\downarrow}(m)$, and the $l\ts{th}$ element of  $\Pi^{\downarrow}_{m\geqslant 1}$ as $\Pi^{\downarrow}_{m\geqslant 1}(l)$. 

\
\begin{thm}\label{optimal-unitary-maximal}
The equivalence class of  unitary operators that maximise the probability of information erasure and, given this constraint, minimise the heat dissipation, is denoted as $[U_\mathrm{opt}(0)]$. This  is characterised by the rules 
\begin{align}
U_\mathrm{opt}(0):\begin{cases}\ket{\psi_n} \mapsto \ket{\varphi_1} \otimes \ket{\xi_m} & \text{if }  p(\psi_n|\rho) = \Pi^{\downarrow}_0(m), \\
\ket{\psi_n} \mapsto \ket{\varphi_l^m}\otimes\ket{\xi_{m}} & \text{if }  p(\psi_n|\rho)= \Pi^{\downarrow}_m(l)  \text{ and }  m\geqslant 1,\\
\end{cases}\label{optimal U appendix}
\end{align}
where, for all $m$, each member of the orthonormal set $\{\ket{\varphi_l^m}\}_l$ is orthogonal to $\ket{\varphi_1}$.  
\end{thm}
\begin{proof}
 The first line conforms with the conditions imposed by \lemref{maximum-probability} and, as such, results in $p(\varphi_1|\rho_\oo')= p_{\varphi_1}^{\max}$. However, here we are restricted to the case  $\ket{\xi_m'}=\ket{\xi_m}$ for all $m$, thereby minimising the contribution to heat dissipation by \corref{corollary-monotone-lattice-superadditive} and \lemref{recursive-ky-fan}. The second line, by virtue of not affecting $p(\varphi_1| \rho_\oo')$, is evidently allowed for a unitary operator in the equivalence class $[U_{\mathrm{maj}}^g]$. This rule takes the $d_\rr$ largest remaining probabilities to states $ \ket{\varphi_l^1}\otimes\ket{\xi_{1}}$, thereby maximising the probability associated with $\ket{\xi_1}$, and so on for the other states $\ket{\xi_m}$. By the same line of reasoning as in \lemref{minimising-heat}, therefore, the contribution to heat dissipation from this line is minimal.   
\end{proof}

\subsection{The tradeoff between probability of information erasure and minimal heat dissipation}\label{proof trade-off relation}

Let us make the following observations:

\begin{enumerate}[(a)]
\item
For any value of $\Delta Q$, $p(\varphi_1| \rho_\oo')$ is maximised when the eigenvectors of $\rho_\oo'$ that have support on $\ket{\varphi_1}$ are given by the set $\{\ket{\varphi_l}\}_l$. This follows from \corref{corollary-monotone-lattice-superadditive}, which implies that  $p(\varphi_1| \rho_\oo')=\sum_l o_l^{'\downarrow} |\<\varphi_1|\varphi_l'\>|^2\leqslant o_1^{'\downarrow}$, where $\rho_\oo'=\sum_l o_l^{'\downarrow}\pr{\varphi_l'} $.
\item For any value of $p(\varphi_1|\rho_\oo')$, $\Delta Q$ is minimised when the eigenvectors of $\rho_\rr'$ are given by the set $\{\ket{\xi_m}\}_m$. This follows from \lemref{recursive-ky-fan}.
\end{enumerate}
  Observations (a) and (b), together, show that in general the optimal case will require that, for all $n$,   $U\ket{\psi_n}=\sum_l\sqrt{\gamma^n_l} \ket{\varphi_l}\otimes \ket{\xi_l^n}$. Here  $\gamma_l^n\geqslant 0$ are the Schmidt coefficients, and  $\ket{ \xi^n_l}= e^{\imag \phi^n_l}\sigma_n \ket{\xi_l}$ with $\sigma_n$  a permutation on the set $\{\ket{\xi_l}\}_l$ and $\phi^n_l \in [0,2\pi)$ a phase. 

Consider now the  algorithm for sequential swaps between 2-dimensional subspaces of $\h_\oo$ and  $\h_\rr$, shown in \fig{sequential swap appendix}.  

\begin{figure}[!h]
\begin{tabular}{|c||c|}\hline
Step (1) & Set $i=2$ and $m=d_\rr$. \\\hline
Step (2) & Sequentially swap  $\ket{\varphi_1}\otimes \ket{\xi_{i}}$  with the vectors $\ket{\varphi_{l}}\otimes \ket{\xi_{m}}$ \\ &  with $l$ running from $d_\oo$ down through to $2$, only if $p_{1,i}< p_{l,m}$. \\\hline
Step (3) & If $m>1$, set $m= m-1$ and go back to Step (2). Else, proceed to Step (4). \\\hline
Step (4) & If $i <d_\rr$, set $i=i+1$, $m=d_\rr$, and go back to Step (2). Else, terminate.  \\\hline
\end{tabular}
\caption{The sequential swap algorithm.}\label{sequential swap appendix}
\end{figure}

During each step of the algorithm, we denote the (updated) probability $p(\varphi_l\otimes \xi_m|U_\mathrm{step} \rho U_\mathrm{step}^\dagger)$ as $p_{l,m}$. Here, $U_\mathrm{step}$ represents the unitary operator that results from conducting the algorithm up to some particular step.

\
\begin{lem}\label{sequential-swap-decreasing-error}
The sequential swap algorithm produces a non-increasing sequence of errors, $\bm{\delta}^\downarrow:=\{\delta_j^{\downarrow}\}_j$, commensurate with a non-decreasing sequence of heat, $\bm{\Delta Q}^{\uparrow}:= \{\Delta Q_j^\uparrow\}_j$, such that the resultant state $\rho_\oo'$ is always passive.

\end{lem}
\begin{proof}
 For every iteration of Step (2), each swap operation increases $p(\varphi_1|\rho_\oo')$, so we obtain the non-increasing sequence of errors $\bm{\delta^\downarrow}$ by construction. Furthermore, each swap increases $p(\xi_i|\rho_\rr')$, while decreasing $p(\xi_m|\rho_\rr')$. To show that this always leads to an increase in heat by \corref{corollary-monotone-lattice-superadditive}, we must show that, for each swap, $i>m$.  Every swap in each iteration of Step (2) effects a permutation on the set $\{p_{1,i},p_{2,m},\dots,p_{d_\oo,m}\}$.   Initially, $p_{1,i}=o_1^\downarrow r_i^\downarrow$. We note that if $o_1^\downarrow r_i^\downarrow < o_l^\downarrow r_m^\downarrow$ with $l\geqslant 2$,  then by necessity $i > m$. As such, the swaps for the first iteration of Step (2), that involve state $\ket{\varphi_1}\otimes \ket{\xi_2}$ and  lead to a permutation in $\{p_{1,2},p_{2,1},\dots,p_{d_\oo,1}\}$, result in a decrease in $p(\xi_1|\rho_\rr')$ and an increase in $p(\xi_2|\rho_\rr')$, which indeed  leads to a non-decreasing sequence of heat.  And so on recursively for all $i$. To show that $\rho_\oo'$ is always passive, we need to show that after each swap,  $p(\varphi_i|U_\mathrm{step} \rho U_\mathrm{step}^\dagger)=\sum_mp_{i,m} \geqslant p(\varphi_j|U_\mathrm{step} \rho U_\mathrm{step}^\dagger)= \sum_m p_{j,m}$ for all $i<j$. This follows from the fact that $\{p_{i,m}\}_i$ are always in non-increasing order, and that every element in $\{p_{i,m}\}_{i\geqslant 2}$ is greater than or equal to all those in $\{p_{i,m'}\}_{i\geqslant 2}$ if $m < m'$.    
\end{proof} 

Now, we wish to show that the non-decreasing sequence of heat $\bm{\Delta Q}^\uparrow$ is optimal for the associated non-increasing sequence of errors $\bm{\delta}^\downarrow$.
 
\  

\begin{thm}\label{sequential-swap-optimal}
If an error $\delta$ can be achieved using the sequential swap algorithm, the consequent heat dissipation will be optimal. Achieving the same $\delta$ with the presence of entanglement in the vectors $\{U_\mathrm{step} \ket{\psi_n}\}_n$ will either increase $\Delta Q$, $\Delta W$, or both.   
\end{thm}
\begin{proof}
By \corref{corollary-monotone-lattice-superadditive}, \lemref{recursive-ky-fan} and \lemref{sequential-swap-decreasing-error}, the heat dissipation due to the sequential swap algorithm is minimal if we are restricted to swap operations. If we are not restricted to performing swap operations, we could also achieve the same error $\delta$ by allowing for entanglement in the vectors $\{U_\mathrm{step} \ket{\psi_n}\}_n$. To show that this  will result in a greater amount of heat dissipation,  it is sufficient to show that doing so would increase  $p_{i,m}$ and decrease $p_{i,m'}$, for some $i$, such that $m>m'$. Likewise, we may show that this would increase the average energy of $\rho_\oo'$, and hence increase $\Delta W$,  by demonstrating that the process would  increase  $p_{i,m}$ and decrease $p_{j,m}$, for some $m$, such that $i>j$. 

Here is a sketch of the proof. First start with $\rho = U_\mathrm{opt}^p(0) \rho {U_\mathrm{opt}^p(0)}^\dagger$, which coincides with  the smallest error $\delta^{\downarrow}_{j_\mathrm{max}}=0$, where $j_\mathrm{max}$ represents  the final step in the swap algorithm.  Here, we have $p_{1,d_\rr}^{j_\mathrm{max}}= p^\downarrow_{d_\rr}$ and $p_{2,1}^{j_\mathrm{max}}= p^\downarrow_{d_\rr+1}$.  The first step of the sequential swap algorithm, run backwards, gives us  $p_{1,d_\rr}^{j_\mathrm{max}-1}= p^\downarrow_{d_\rr+1} $ and $p_{2,1}^{{j_\mathrm{max}}-1}= p^\downarrow_{d_\rr}$, with $\delta^{\downarrow}_{j_\mathrm{max}-1}= p^\downarrow_{d_\rr}- p^\downarrow_{d_\rr+1}$. All other values are the same as before.  Now instead of the swap operation, have 
\begin{equation}
U_{j_\mathrm{max}-1}\ket{\psi_{d_\rr}}= \sqrt{\gamma} \ket{\varphi_1} \otimes \ket{\xi_{d_\rr}} + \sqrt{1-\gamma} \ket{\varphi_i}\otimes \ket{\xi_m},
\end{equation} and 
\begin{equation}
U_{j_\mathrm{max}-1}\ket{\psi_{d_\rr+(m-1)(d_\oo-1) +i}}= \sqrt{1-\gamma} \ket{\varphi_1} \otimes \ket{\xi_{d_\rr}} - \sqrt{\gamma} \ket{\varphi_i}\otimes \ket{\xi_m},
\end{equation} with all other $U_{j_\mathrm{max}-1}\ket{\psi_n}$ defined by $U_\mathrm{opt}^p(0)$. With some choice of $\gamma,i,m$, we can again obtain 
\begin{equation}
p_{1,d_\rr}^{j_\mathrm{max}-1}= \gamma p^\downarrow_{d_\rr} + (1-\gamma) p^\downarrow_{d_\rr+(m-1)(d_\oo-1)+i}= p^\downarrow_{d_\rr+1},
\end{equation} and hence the same value of $\delta^{\downarrow}_{j_\mathrm{max}-1}$ as before. This, however, will lead to     
\begin{equation}
p_{2,1}^{j_\mathrm{max}-1}=   p^\downarrow_{d_\rr+1}\leqslant p^\downarrow_{d_\rr} ,
\end{equation} and 
\begin{equation}
p_{i,m}^{j_\mathrm{max}-1}= (1-\gamma )p^\downarrow_{d_\rr} + \gamma p^\downarrow_{d_\rr+(m-1)(d_\oo-1)+i} \geqslant p^\downarrow_{d_\rr+(m-1)(d_\oo-1)+i}.
\end{equation}  
In other words, using the new entangling unitary operator, instead of the sequential swap algorithm, will result in $p_{2,1}$ to decrease, and $p_{i,m}$ to increase.  If $i=2$ and $m\geqslant 2$, this will result in a larger $\Delta Q^\uparrow_{j_\mathrm{max}-1}$. Conversely, if $m=1$ and $i\geqslant 3$, this will increase the average energy of  the object, and thereby increase $\Delta W$.  If both $i\geqslant 3$ and $m\geqslant 2$, then both $\Delta Q$ and $\Delta W$ will be larger.   The same line of reasoning would apply for entanglement of higher Schmidt-rank. 
\end{proof}
    
\subsection{Full erasure of a maximally mixed qudit with a  harmonic oscillator}\label{qudit erasure calculations appendix} 

\begin{figure}[!htb]
\centering
\subfigure[Decreasing $\omega$ and increasing $d$: constant $\| H_\rr \|$.]{\label{reset-Landauer-small-freq-norm-constant}
\includegraphics[width=8 cm]{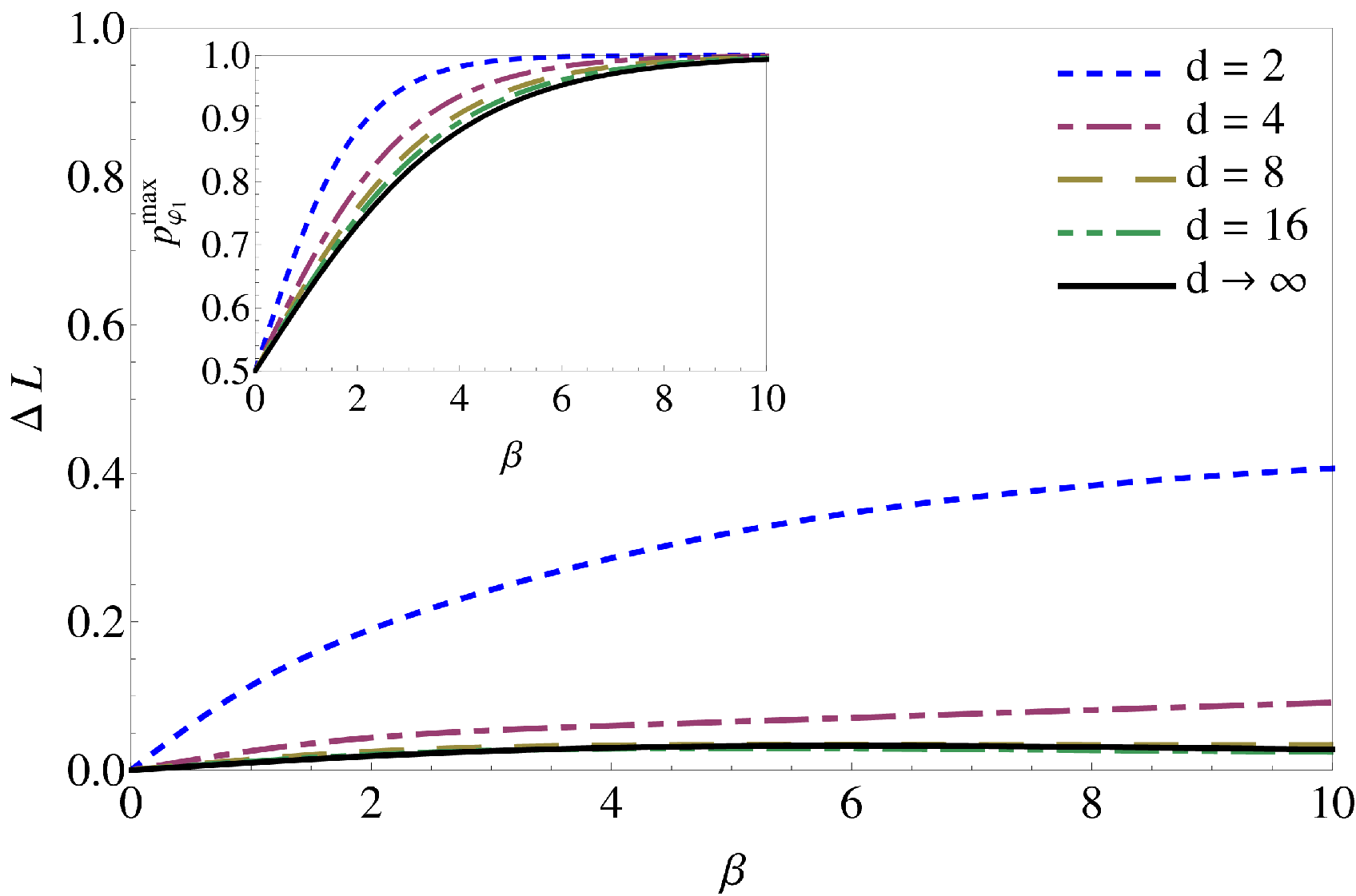}}
\subfigure[Decreasing $\omega$ and increasing $d$: increasing $\| H_\rr \|$]{\label{reset-Landauer-small-freq-norm-increase}
\includegraphics[width=8 cm]{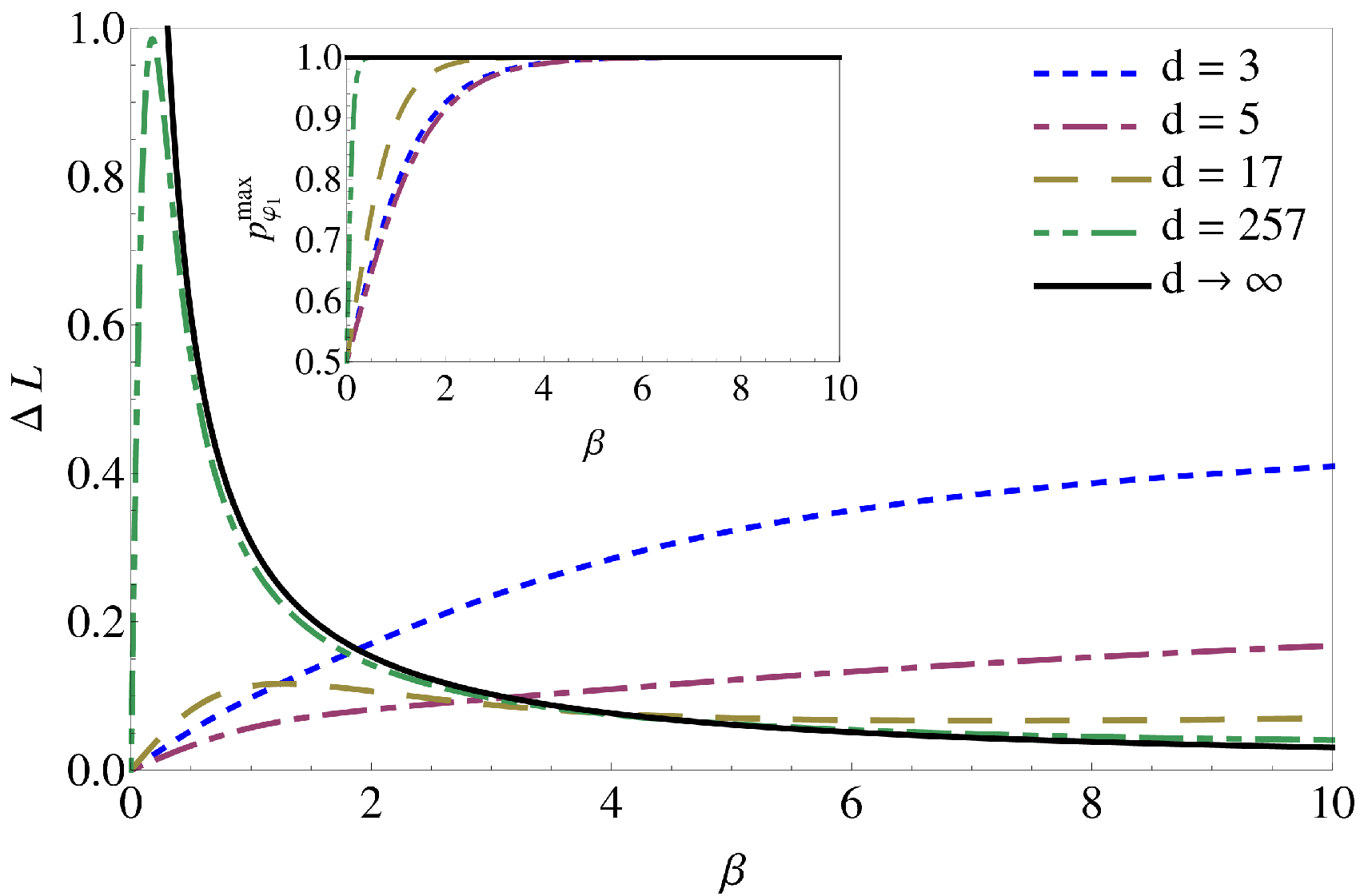}}
\caption{ (Colour online)  Comparison between two different methods of taking the double limit of $d \to \infty, \omega \to 0$, and their effect on  $p^{\max}_{\varphi_1}$ and $\Delta L$, when $d_\oo=2$  (a)  Here the frequencies are  $\omega=1/(d-1)$. Therefore the Hamiltonian norm is  $1$ for all $d$.  For any given $\beta$, as $d$ grows larger, thereby making $\omega$ smaller, both $\Delta L$ and $p^{\max}_{\varphi_1}$ decrease. (b) Here $d=2^n+1$  for frequencies  $\omega=1/n$ with $n \in \mathds{N}$. Therefore the Hamiltonian norm is  $2^n/n$.  For a sufficiently large $\beta$, as $n$ grows larger, thereby making $\omega$ smaller,  $\Delta L$ decreases while $p^{\max}_{\varphi_1}$ increases.   }
\end{figure}

Here, we expound on the example of using a ladder system as a reservoir, but consider what happens as we take the limit of infinitely large $d$. In this limit we may call the ladder system a harmonic oscillator. Furthermore, we consider the object as a qudit, with Hilbert space $\h_\oo \simeq \co^{d_\oo}$, prepared in the maximally mixed state
\begin{equation}
\rho_\oo= \frac{1}{d_\oo}\sum_{l=1}^{d_\oo} \pr{\varphi_l}.
\end{equation}
Consider a harmonic oscillator of frequency $\omega$, with the ground state energy, $\lambda_1^\uparrow$, defined as zero. As such, the $m$\ts{th} smallest energy is $\lambda^\uparrow_m= \omega(m-1)$. Given a fixed and finite $\omega$,  in the limit as $d$ tends to infinity there will be infinitely many eigenvalues of $H_\rr$ that become formally infinite, and hence infinitely many probabilities $r_m^\downarrow$ vanish. As such, we have
\begin{equation}
\lim_{d \to \infty}\rho'=\pr{\varphi_1}\otimes \rho_\rr',
\end{equation}
whereby $p_{\varphi_1}^\mathrm{max}=1$. In addition, 
\begin{equation}
\lim_{d \to \infty} \rho_\rr'= \sum_{m=1}^\infty \frac{r^\downarrow_m}{d_\oo}\sum_{j=0}^{d_\oo-1}\left(\pr{\xi_{d_\oo m-j}} \right),
\end{equation}
and a resulting heat dissipation of
\begin{align}
 \lim_{d \to \infty}\Delta Q&=\sum_{m=1}^\infty \frac{r^\downarrow_m}{d_\oo}\sum_{j=1}^{d_\oo-1}(\lambda^\uparrow_{d_\oo m-j})- \sum_{m=1}^\infty r^\downarrow_m \lambda_m^\uparrow, \nonumber \\ &=  \frac{\omega(d_\oo-1)}{2}\coth\left(\frac{\beta \omega}{2}\right)> \frac{(d_\oo-1)}{\beta}.
\end{align}
 $\Delta Q$ approaches $(d_\oo-1)k_B T$ in the limit as $\omega$ becomes vanishingly small, and hence the optimal case is achieved when we take the double limit of  $d$ going to infinity while $\omega$ goes to zero. 

Of course, the ``\emph{rate}'' at which we take the limit $d\to \infty$ must  be greater than that at which $\omega$ approaches zero.     As shown in \fig{reset-Landauer-small-freq-norm-constant}, for the case of $d_\oo=2$,  if we increase $d$ while decreasing $\omega$ in such a way so as to keep $\|H_\rr \|$ constant,  both the probability of qubit erasure and the heat dissipation  decrease. Precisely, this may be achieved if we define the frequency as $\omega := \| H_\rr \|/(d-1)$.
 In the limit as $d$ tends to infinity and $\omega$ vanishes, the spectra of $H_\rr$ and $\rho_\rr(\beta)$ become  continuous. That is to say, $\lambda_{m+1}^\uparrow- \lambda_m^\uparrow \to 0$ and $r_{m}^\downarrow- r_{m+1}^\downarrow \to 0$, for all $m$. We may therefore simplify our calculations by replacing sums with Riemann integrals. First, we note that in this case
the maximum probability of qudit erasure is\begin{equation}
\lim_{\substack{\omega \to 0\\d \to \infty}}p_{\varphi_1}^{\max} =\frac{\int_0^{\frac{\|H_\rr\|}{d_\oo}}e^{-\beta x}dx}{\int_0^{\|H_\rr\|}e^{-\beta x}dx}= \frac{1}{\sum_{j=0}^{d_\oo-1} e^{-\frac{\beta j\|H_\rr \|}{d_\oo}}}.\label{maximum probability case 2}
\end{equation} 
 Moreover the heat dissipation is 
\begin{align}
\lim_{\substack{\omega \to 0 \\d \to \infty}} \Delta Q &= \lim_{\substack{\omega \to 0\\ d \to \infty}}\frac{1}{d_\oo}\sum_{m=1}^{d/d_\oo} \left(\sum_{j=0}^{d_\oo-1}r^\downarrow_{m+jd/d_\oo}\right)\left(\sum_{j=0}^{d_\oo-1}\lambda^\uparrow_{d_\oo m - j}\right) -\lim_{\substack{\omega \to 0\\ d \to \infty}}\sum_{m=1}^d r^\downarrow_m \lambda^\uparrow_m,  \nonumber \\  &= 
\frac{\frac{1}{d_\oo}\int_0^{\|H_\rr\|}\left(\sum_{j=0}^{d_\oo-1}e^{-\beta (\frac{x+j\|H_\rr\|}{d_\oo})}\right)x \ dx}{\int_0^{\|H_\rr\|}e^{-\beta x} \ dx} -\frac{\int_0^{\|H_\rr\|}e^{-\beta x} x \ dx}{\int_0^{\|H_\rr\|}e^{-\beta x} \ dx}, \nonumber \\
&= \frac{d_\oo-1}{\beta} + \frac{\|H_\rr\|}{2}\left[\coth\left(\frac{\beta \| H_\rr\|}{2}\right) - \coth\left(\frac{\beta \| H_\rr\|}{2 d_\oo}\right) \right].\label{maximum heat case 2}
\end{align}
These functions take the values of $1$ and $(d_\oo-1)k_B T$, respectively, precisely when $\|H_\rr\|$ is infinitely large. Therefore, if $\omega$ and $d$ decrease and increase, respectively, in such a way so that $\|H_\rr\|$ also increases, then in this limit we achieve the optimal case of full information erasure with the minimal heat dissipation of $(d_\oo-1)k_B T$. One way of ensuring this,  as shown in \fig{reset-Landauer-small-freq-norm-increase}, is to  define the dimension of the reservoir as $d=2^n+1$,  where $n$ is a natural number, while defining the frequency as $\omega=\thickbar \omega/n$. The Hamiltonian norm will be 
\begin{equation}
\| H_\rr\|= \thickbar \omega \frac{2^n}{n}, 
\end{equation}     
 which, in the limit as $n$ tends to infinity, becomes infinitely large. 
 
\subsubsection{Full Erasure of a qubit with an initial bias} \label{general qubit erasure appendix}

We have shown that when the whole harmonic oscillator is used as a reservoir we can  fully  purify a qubit in a maximally mixed state, where the entropy reduction is $\Delta S=\log(2)$, with a heat cost of $\Delta Q > k_B T$. Here we wish to evaluate the optimal $\Delta Q$ for arbitrary initial states of the qubit and, hence, arbitrary entropy changes $\Delta S$. To this end, define the initial state of the object as
\begin{equation}
\rho_\oo = q \pr{\varphi_1}+(1-q) \pr{\varphi_2} \ , \ q \in \left[\frac{1}{2},1\right).
\end{equation} 
The non-increasing vector of probabilities $\bm{p^\downarrow}$ can therefore be written  as
\begin{equation}
\bm{p^\downarrow}= \{ q r^\downarrow_1,\dots, q r^\downarrow_k, (1-q) r^\downarrow_1, \dots, (1-q) r^\downarrow_k,\dots\}, 
\end{equation}  
where the ordering implies that 
\begin{equation}
\frac{q}{1-q}\geqslant \frac{r^\downarrow_1}{r^\downarrow_k}=e^{\beta \omega(k-1)}.
\end{equation}

After the joint evolution with an infinite-dimensional reservoir, the above sequence $\bm{p^\downarrow}$ describes the spectrum of $\rho_\rr'$, with the first entry associated with eigenvector $\ket{\xi_1}$, and so on.   In the limit of infinitesimally small $\omega$, the energy spectrum of the reservoir and, hence, the probabilities $\bm{r^\downarrow}$ can be approximated as a continuum. We may therefore evaluate $\Delta Q$ by
\begin{align}
\Delta Q &=\frac{ \sum_{n=1}^\infty Q(n)- \int_0^\infty x e^{-\beta x}dx}{\int_0^{\infty}e^{-\beta x} \ dx}= \frac{2q(1-q)\log(\frac{q}{1-q})}{\beta(2q -1)}, \nonumber \\
Q(n)&=q \int_{(2n-2)\Omega}^{(2n-1)\Omega}xe^{-\beta(x-(n-1)\Omega)}dx+(1-q)\int_{(2n-1)\Omega}^{2n \Omega}xe^{-\beta(x-n \Omega)}dx,
\end{align}
where $\Omega$ is the energy ``\emph{width}'' which satisfies $q/(1-q)=e^{\beta \Omega}$. In the limit as $q$ tends to one-half, $\Delta Q$ approaches $1/\beta$ as in our previous analysis.

\subsection{Object as a component of a thermal system}\label{Object as a component of a thermal system proof}

In this case, the heat dissipation is 
\begin{align}
\Delta Q:= \tr[H(U\rho(\beta) U^\dagger-\rho(\beta))]&= \frac{1}{\beta}S(U\rho(\beta) U^\dagger\| \rho(\beta)), \nonumber \\ &= \frac{1}{\beta} \left(\sum_{n=1}^{d_\oo d_\kk} q^U_n \log\left(\frac{1}{p_n^\downarrow}\right)-S(\rho')\right),\nonumber \\
&= \frac{1}{\beta} \left(\bm{q^U} \cdot \bm{\mathrm{log}_p^\uparrow} - S(\rho(\beta))\right),\label{relative entropy term appendix}
\end{align}
where $\bm{q^U}$ is a vector of real numbers
\begin{equation}
q^U_n:=\sum_{m=1}^{d_\oo d_\kk}p_m^\downarrow|\<\xi_n|U|\xi_m\>|^2,\label{relative entropy term probabilities appendix}
\end{equation}
and $\bm{\mathrm{log}_p^\uparrow}:= \{\log(1/p_n^\downarrow)\}_n$ a non-decreasing vector. We now determine the properties that $U\in [U_\mathrm{maj}]$ must satisfy so as to minimise $\Delta Q$ conditional on maximising $p(\Psi|\rho_\oo')$.

\
\begin{prop}\label{optimal unitary global thermal}
For a fixed Hamiltonian, $\Delta Q$ given maximally probable information erasure is minimised by choosing $U$ from an equivalence class of unitary operators $[U_\mathrm{maj}^1] \subset [U_\mathrm{maj}]$ such that $\bm{{q^{U_\mathrm{maj}^1}}^\downarrow}= \bm{{q^{U_\mathrm{maj}^1}}}$ and  $\bm{{q^{U_\mathrm{maj}^1}}^\downarrow }\succ \bm{{q^{U}}}^\downarrow$ for all $U \in [U_\mathrm{maj}]$.
\end{prop}
\begin{proof}
As $S(\rho')= S(\rho(\beta))$ and $\bm{\mathrm{log}_p^\uparrow}$ are fixed by the initial conditions, then $\Delta Q$ is minimised by minimising $\bm{q^U}\cdot \bm{\mathrm{log}_p^\uparrow}$.  This is achieved by $U_\mathrm{maj}^1$ as a consequence of \thmref{monotone-lattice-superadditive}   and \corref{corollary-monotone-lattice-superadditive}. \end{proof}
 
Of course, we may also minimise $\Delta Q$ by engineering the Hamiltonian itself.

\
\begin{prop} \label{optimal Hamiltonian global thermal}
$\Delta Q$ given maximally probable information erasure will be minimised if all  $\ket{\xi_n}$ that have support on $\{\ket{\Psi} \otimes \ket{\phi_j}\}_j$ are given from the  set $\{\ket{\Psi} \otimes \ket{\phi_j}\}_j$. 
\end{prop}
\begin{proof}
As shown during the proof of the Klein inequality in \citep{Sagawa-relative-entropy}, given a constant spectrum of $\rho$ and $\sigma$,  $S(\rho \| \sigma )$ is minimised when $\rho$ commutes with $\sigma$. Since  $\Delta Q$  takes its smallest value by minimising $S(U\rho(\beta) U^\dagger\| \rho(\beta))$, to achieve this $U\rho(\beta) U^\dagger$ must commute with $\rho(\beta)$.  By construction, $U_\mathrm{maj}^1\ket{\xi_n}=\ket{\Psi}\otimes \ket{\phi_j}$ for all $n\in \{1,\dots,d_\kk\}$. So, if $|\<\xi_m|U_\mathrm{maj}^1 |\xi_n\>|>0$, to minimise $\Delta Q$ we must have $\ket{\xi_m} \in \{\ket{\Psi} \otimes \ket{\phi_j}\}_j$.
\end{proof}

%\bibliography{references}

\bibliographystyle{ieeetr}

\end{document}